\documentclass[12pt]{amsart}
\usepackage{epsfig}
\usepackage{amsfonts}
\usepackage{amsmath}
\newtheorem{thm}{Theorem}[section]

\newtheorem{cor}[thm]{Corollary}
\newtheorem{lem}[thm]{Lemma}
\newtheorem{prop}[thm]{Proposition}
\newtheorem*{prob*}{Problem}
\newtheorem*{thm*}{Theorem}

\theoremstyle{definition}

\newtheorem*{defn*}{Definition}
\newtheorem{rem}[thm]{Remark}
\newtheorem*{Remarks}{Remarks}
\newtheorem*{rem*}{Remark}
\numberwithin{equation}{section}

\newcommand{\C}{\mathbb C}

\newcommand{\R}{\mathbb R}

\newcommand{\X}{\mathfrak{X}}

\newcommand{\Pc}{\mathcal{P}}

\DeclareMathOperator{\const}{const}

\DeclareMathOperator{\E}{\mathbb{E}}
\DeclareMathOperator{\Ginibre}{Ginibre}
\DeclareMathOperator{\interpol}{interpol}

\newcommand{\re}{\mathop{\mathrm{Re}}}

\newcommand{\Tr}{\mathop{\mathrm{Tr}}}

\begin{document}
\title[Product matrix processes for  multi-matrix models]
 {\bf{Product matrix processes for coupled multi-matrix models and their hard edge scaling limits}}

\author{Gernot Akemann}
\address{Faculty of Physics,  Bielefeld University,  
P.O. Box 100131, 33501 Bielefeld, Germany
} \email{akemann@physik.uni-bielefeld.de}

\author{Eugene Strahov}
\address{Department of Mathematics, The Hebrew University of
Jerusalem, Givat Ram, Jerusalem
91904, Israel}\email{strahov@math.huji.ac.il}


\commby{}
\begin{abstract}

Product matrix processes are multi-level point processes formed by the singular values of random matrix products.
In this paper we study such processes where the products of up to $m$ complex random matrices are no longer independent,
by introducing a coupling term and potentials for each product.
We show that such a process
still forms a multi-level determinantal point processes,
and  give formulae for the relevant correlation functions in terms of the corresponding kernels.

For a special choice of potential, leading to a Gaussian coupling between the $m$th matrix 
and the product of all previous $m-1$ matrices,
we derive a contour integral representation for the correlation kernels suitable for an asymptotic analysis of large matrix size $n$.
Here, the correlations between the first $m-1$ levels equal that of the product of $m-1$ independent matrices,
whereas all correlations with the $m$th level are modified.
In the hard edge scaling limit at the origin of the spectra of all products
we find three different asymptotic regimes.
The first regime corresponding to weak coupling agrees with the multi-level process for the product of $m$
independent complex Gaussian matrices for all levels, including the $m$-th. This process was introduced by one of the authors and
can  be understood as a multi-level extension of the Meijer $G$-kernel introduced by Kuijlaars and Zhang.
In the second asymptotic regime at strong coupling the point process on level $m$ collapses onto level $m-1$,
thus leading to the process of $m-1$ independent matrices.
Finally, in an intermediate regime where the coupling is proportional to $n^{\frac12}$,
we obtain a family of parameter dependent kernels, 
interpolating between the limiting processes in the weak and strong coupling regime.
These findings generalise previous results of the authors and their coworkers for $m=2$.
\end{abstract}

\maketitle
\tableofcontents
\section{Introduction}\label{SectionIntroduction}
By a random \textit{multi-matrix model} one usually means a probability measure defined on a space  formed by a collection of matrices
of the same type. Probably the best known example is that of Hermitian matrices coupled in a chain
considered by Eynard and Mehta \cite{EynardMehta}. This example represents the simplest case of multi-matrix models known from applications in quantum field theory.
Another class of examples consists of multi-matrix models of positive Hermitian matrices subject to the Cauchy interaction. This class of multi-matrix models
was introduced in  Bertola, Gekhtman, and  Szmigielski \cite{BertolaGekhtmanSzmigielski}, and studied
further in  Bertola, Gekhtman, and  Szmigielski \cite{BertolaGekhtmanSzmigielski1,BertolaGekhtmanSzmigielski2}, and in
 Bertola and  Bothner \cite{BertolaBothner}. For other examples of multi-matrix models, and for an explanation of their relevance
 to quantum field theory and to statistical mechanics we refer the reader to Eynard,  Kimura, and Ribault \cite[Section 2.2.]{EynardKimuraRibault},
 Filev and O'Connor \cite{FilevO'Connor}, Bertola,  Eynard,  Harnad \cite{BertolaEynardHarnad}, and references therein.

 The starting point of the present work is the observation that the problem about the distribution of singular values for a product of independent complex
 Gaussian matrices can be reformulated in terms of a multi-matrix model. This multi-matrix model can be defined by the probability measure
 \begin{equation}\label{MultiMatrixGinibreProbabilityMeasure}
 \begin{split}
\frac{1}{Z_n}&e^{-\sum\limits_{l=1}^m\Tr\left[G_l^*G_l\right]}\prod\limits_{l=1}^mdG_l.
\end{split}
\end{equation}
over the set of rectangular complex matrices $\left(G_1,\ldots,G_m\right)$, where $G_l$ is of size
$\left(n+\nu_l\right)\times \left(n+\nu_{l-1}\right)$, $\nu_0=0$,  $\nu_1\geq 0$, $\ldots$, $\nu_{m-1}\geq 0$,
$dG_l$ is the corresponding flat complex Lebesgue measure, and $Z_n$ is a normalisation constant.
Given this probability measure one can ask about the distribution of complex eigenvalues and of squared singular values of the total product matrix
$Y_m=G_m\cdots G_1$. It turns out that the eigenvalues of $Y_m$ form a determinantal point process on the complex plane
which can be understood as a generalisation of the classical Ginibre ensemble. This fact was first proved in Akemann and Burda \cite{Akemann1}, see
Adhikari,  Kishore Reddy,  Ram Reddy, and  Saha \cite{Adhikari}, Forrester \cite{Forrester}, Ipsen \cite{Ipsen},
 Akemann,  Ipsen, and Strahov \cite{AkemannIpsenStrahov}, Forrester and Ipsen \cite{ForresterIpsen} for different proofs and extensions of this result.
Moreover, it was shown in Akemann,  Kieburg, and   Wei \cite{AkemannKieburgWei}, and in Akemann,  Ipsen, and Kieburg \cite{AkemannIpsenKieburg}
that the joint probability density function of the squared singular values of $Y_m$ forms a
determinantal point process, representing a special
polynomial ensemble, a notion introduced later by Kuijlaars and Stivigny \cite{ArnoDries}.
A contour integral representation for the correlation kernel of this ensemble was derived in Kuijlaars and Zhang \cite{KuijlaarsZhang},
which enabled a detailed analysis of different scaling limits, see   Kuijlaars and Zhang \cite{KuijlaarsZhang},   Liu,  Wang, and Zhang \cite{LiuWangZhang}.

 A natural (multi-level) generalisation of the point process formed by the squared singular values of  the total product matrix
$Y_m$ can be constructed as follows. For each $l\in\left\{1,\ldots,m\right\}$, denote by $y_j^l$, $j=1,\ldots,n$,
the squared singular values of the partial product $Y_l=G_l\cdots G_1$, that is the eigenvalues of $Y_l^*Y_l$. The configuration
\begin{equation}\label{configuration}
\left\{\left(l,y_j^l\right)\vert l=1,\ldots,m;j=1,\ldots,n\right\}
\end{equation}
of all eigenvalues
forms a point process on $\left\{1,\ldots,m\right\}\times\R_{>0}$. This process was introduced and  studied in Strahov \cite{StrahovD}, calling it the \textit{Ginibre product process}. It was shown in \cite{StrahovD} that it is a multi-level determinantal point process. Furthermore, a contour integral representation for the correlation kernel of this process was derived, and its hard edge scaling limit computed.

Here, we  will drop the assumption that the matrices $G_1$, $\ldots$, $G_m$ are independent as in (\ref{MultiMatrixGinibreProbabilityMeasure}),
but assume instead that  these or their products becomes coupled. We can still  construct and investigate the multi-level point process formed by configurations (\ref{configuration}) and refer to it as the \textit{product matrix process} associated with
$G_1$, $\ldots$, $G_m$. It is an interesting general problem to describe statistical properties of such product matrix processes,
their relevant correlation functions, and their scaling limits. Of course it is desirable that it still forms a multi-level determinantal point process, with explicit formulae for the correlation kernel and thus the correlation functions.
An example for such a setup is the multi-matrix model for Hurwitz numbers \cite{AmbjornChekov}. In this paper we introduce and study a different multi-matrix model of statistically dependent random matrices, also satisfying these requirements. We show  (see Theorem \ref{TheoremMostGerneralCorrelationKernel}) that these product matrix processes generalise the
Ginibre product process studied in Strahov \cite{StrahovD}.
For a particular choice of potentials
the first $m-1$ levels remain that of products of $m-1$ independent random matrices, to which the $m$-th level is coupled.
For this example we derive contour integral representations for the relevant correlation kernels,
see Theorem \ref{TheoremDoubleIntegralRepresentationExactKernel}.

In the hard edge scaling limit  at large matrix sizes we distinguish three different asymptotic regimes.  In the first regime corresponding to
 weak coupling,  the limiting  product matrix process is the same as for the product of $m$
independent complex Gaussian matrices, as described in Strahov \cite{StrahovD}. This process can be called the $m$-level Meijer $G$-kernel process, a multi-level extension of the Meijer $G$-kernel process introduced by Kuijlaars and Zhang \cite{KuijlaarsZhang}.
In the second asymptotic regime at strong coupling the $m$-th level collapses to the point process on the $(m-1)$-th level.
In other words, the correlations between the $m$-th level and the levels up to $m-2$ are the same as the respective correlations between the $(m-1)$-th and the levels up to $m-2$.
The correlations between level $m$ and $m-1$ are as if they are on the same level, except at colliding points where we find a contact interaction in form of a Dirac delta function. Finally, at intermediate coupling we obtain a  limiting process of $m$ levels, \textit{interpolating} between that in the weak coupling regime and that in the strong coupling regime.
All three hard edge scaling limits described above are given by Theorem \ref{TheoremHardEdgeGinibreCouplingProcess}
that we consider as the main achievement of the present work.

We note that interpolating ensembles are of great interest in Random Matrix Theory. In the literature one can find examples interpolating between Gaussian ensembles of different symmetry classes - one classical example is that of Pandey and Mehta \cite{PandeyMehta}
that interpolates between the Gaussian Orthogonal and Gaussian Unitary Ensemble. Determinantal processes whose edge behaviour
interpolates between the Poisson process and the Airy process
were considered in Moshe,  Neuberger, and Shapiro \cite{Moshe}, and studied
further by Johansson \cite{JohanssonGumbel}. In the context of products of random matrices, a determinantal process with $m=2$
interpolating between the classical Bessel-kernel process (at $m=1$), and  the Meijer $G$-kernel process for the product of $m=2$
independent Gaussian matrices was obtained by the authors in \cite{AkemannStrahov,AkemannStrahov1}, and further extended most recently by Liu \cite{DZLm2} and a joint work \cite{ACLS}.
However, we are not aware of further examples of multi-level interpolating ensembles, with properties described
by Theorem \ref{TheoremInterpolation}.
In particular, the results mentioned above lead to three different scaling limits for the biorthogonal ensemble
formed by the squared singular values of the total product matrix $Y_m$: that of the Meijer $G$-kernel process for the product of
  $m$ ($m-1$) independent rectangular matrices with complex Gaussian entries at weak (strong) coupling, and  a
  determinantal process interpolating  between these correlation kernels,
  see Corollary \ref{CorollaryTheoremHardEdgeGinibreCouplingProcess}.

 This paper is organised as follows. In Section  \ref{SectionExactResults} we present exact results at finite matrix sizes.
 In particular, in Section  \ref{SectionExactResults} we define a family of multi-matrix models,
 and introduce the product matrix process associated with this family. We show that the product matrix process under considerations
 is a multi-level determinantal process, and give formulae for the correlation kernels.
 In Section \ref{SectionHardEdgeScalingLimits} we compute the hard edge scaling limits in our example,
 corresponding to different asymptotic regimes, and in Section \ref{SectionInterpolatingProcess} we describe the properties
 of the obtained interpolating multi-level determinantal process. Sections \ref{SectionIntegrationFormula}-\ref{SectionProofTheoremInterpolation}
 contain the proofs our statements.\\[2ex]
 
 \textbf{Acknowledgements.} We are very grateful to Marco Bertola, Tomasz Checinski and Mario Kieburg for discussions and useful comments. The anonymous referee is also thanked for several comments and corrections.
This work was supported by the 
DFG through grant AK35/2-1 
and CRC 1283 "Taming uncertainty and profiting from randomness 
and low regularity in analysis, stochastics and their applications"
(G.A.).

\newpage

 \section{Exact results for general coupling and finite matrix size}\label{SectionExactResults}
 \subsection{Coupled multi-matrix models with general potential functions}\label{SectionExactResults1}
 Fix $a>0$, $b>0$, and  consider a multi-matrix model defined by the probability distribution
 \begin{equation}\label{MainProbabilityMeasure}
 \begin{split}
&\frac{1}{Z_n}
\exp\left[-a\sum\limits_{l=1}^m\Tr\left[G_l^*G_l\right]+b\Tr\left[G_m\cdots G_1+\left(G_m\cdots G_1\right)^*\right]
\right]
\\
&\times
\exp\left[
-\sum\limits_{l=1}^m\Tr \left[V_l\left(\left(G_{l}\cdots G_1\right)^*\left(G_{l}\cdots G_1\right)\right)\right]\right]
\prod\limits_{l=1}^mdG_l,
\end{split}
\end{equation}
normalised by the constant $Z_n$. It depends on $m$
rectangular complex matrices 
$G_l$ is  of size
$\left(n+\nu_l\right)\times \left(n+\nu_{l-1}\right)$ for $l=1,\ldots,m$, with
$\nu_0=0$, $\nu_m=0$, $\nu_1\geq 0$, $\ldots$, $\nu_{m-1}\geq 0$,
and where $dG_l$ is the corresponding flat complex Lebesgue measure. Here,  $V_l$ (the potentials) are some scalar positive functions
which are continuous and grow fast enough at infinity to guarantee the convergence of the corresponding matrix measure (\ref{MainProbabilityMeasure}).

Note that the potentials and in particular the
parameter $b$ introduce a \textit{coupling} amongst the matrices: if $b=0$, and $V_l=0$ $\forall l=1,\ldots,m$, then  the matrices $G_1$, $\ldots$, $G_m$ become independent
Gaussian random matrices.

The multi-level point process associated with the multi-matrix model (\ref{MainProbabilityMeasure}) of coupled random matrices 
can be constructed 
as that for the Ginibre product process \cite{StrahovD} of independent matrices.
Namely
denote by $\{y_j^l\}_{j=1,\ldots,n}$ the set of squared singular values of matrix $Y_l=G_l\cdots G_1$, these are the eigenvalues of $Y_l^*Y_l$. The configuration of all these eigenvalues
forms a point process on $\left\{1,\ldots,m\right\}\times\R_{>0}$. We will call this point process the \textit{product matrix
process} corresponding to the multi-matrix model (\ref{MainProbabilityMeasure}).

Let us indicate some particular cases of the multi-matrix model (\ref{MainProbabilityMeasure}) studied previously.
At zero coupling $b=0$, and $V_1=\ldots=V_m=0$,  we are dealing with $m$ independent complex Ginibre matrices, and
the product matrix process turns into the Ginibre product process studied in \cite{StrahovD} (by rescaling all matrices $G_l$ we can set $a=1$).
A second example with $b=0$,  $V_1=\ldots=V_{m-1}=0$ and non-vanishing potential $V_m$ was considered in \cite{AmbjornChekov} in order to study so-called hypergeometric Hurwitz numbers. As  a further feature a non-trivial covariance matrix $\Sigma^{-1}$ is introduced there in the measure for the first matrix, replacing $\Tr [G_1^*G_1]$ by $\Tr[\Sigma^{-1}G_1^*G_1]$. While in \cite{AmbjornChekov} loop equations were applied we will indicate below that this model defines a multi-level point process.

For $m=2$ matrices with non vanishing coupling $b\neq0$,
with both potential functions equal to zero, $V_1=V_2=0$,
and parameters given by $a=\frac{1+\mu}{2\mu}$, $b=\frac{1-\mu}{2\mu}$, we
obtain a model of two coupled matrices introduced by Osborn \cite{Osborn} in the context of QCD with a baryon chemical potential $\mu$.
For this model, the distribution of squared singular values of the product matrix $Y_2$ was studied by the authors in \cite{AkemannStrahov,AkemannStrahov1}. By its construction in \cite{Osborn},
the two coupled matrices can be defined in terms of linear combinations of independent Gaussian matrices, too.
Namely, let $A$, $B$ be two independent matrices of sizes $n\times L$, with i.i.d. standard complex Gaussian entries.
Define random matrices $G_1$ and $G_2$ by
\begin{equation}
G_1=\frac{1}{\sqrt{2}}\left(A-i\sqrt{\mu}B\right),\;\; G_2=\frac{1}{\sqrt{2}}\left(A^*-i\sqrt{\mu}B^*\right).
\label{m=2-model}
\end{equation}
Assume that $L\geq n$. Then the joint distribution of $G_1$ and $G_2$ is given by probability measure (\ref{MainProbabilityMeasure})
with  $a=\frac{1+\mu}{2\mu}$ and $b=\frac{1-\mu}{2\mu}$ for $0<\mu<1$, $m=2$, $V_1=V_2=0$, and $\nu_1=L-n$. The same model with $m=2$ and $V_1=V_2=0$ was generalised in \cite{DZLm2,ACLS} to include non-trivial covariance matrices instead of the scalar parameters $a$ and $b$. Here, again only the squared singular values of the product matrix $Y_2$ were studied.

Without loss of generality we can always assume that the parameter $a$ in formula (\ref{MainProbabilityMeasure}) can be set to $a=1$.
Indeed, this can be achieved by defining
new matrices $\widetilde{G_1}$, $\ldots$, $\widetilde{G_m}$ given by
$$
\widetilde{G_l}=\sqrt{a}G_l,\;\; 1\leq l\leq m,
$$
together with  a replacement of the parameter $b$ by ${b}\,{a^{-\frac{m}{2}}}$, and a modification of the potential functions accordingly.
From now on, we restrict our considerations to the probability measures
defined by expression (\ref{MainProbabilityMeasure}) with $a=1$.
\subsection{Density of product matrix processes and related biorthogonal ensembles}
By definition, the density of the product matrix process corresponding to the multi-matrix model (\ref{MainProbabilityMeasure}) is that of the vector
$$
\underline{y}=\left(y^m,\ldots,y^1\right)\in\left(\R_{>0}^n\right)^m,
$$
where $y^l$ is the vector of squared singular values of $G_l\cdots G_1$.
Our first results gives the joint probability density  of $\underline{y}$ explicitly.
\begin{thm}\label{TheoremGinibreCouplingDensity}
Consider a multi-matrix model defined by the probability measure (\ref{MainProbabilityMeasure}), with $a=1$.
Denote by $\{y_j^l\}_{j=1,\ldots,n}$ the set of squared singular values of the matrix
$Y_l=G_l\cdots G_1$ for each $l=1,\ldots,m$.
Thus for each $l$  the vector $y^l=\left(y^l_1,\ldots, y^l_n\right)$ contains the
eigenvalues of the matrix $Y_l^* Y_l$. The joint probability density of $\underline{y}=\left(y^m,\ldots,y^1\right)$ reads
\begin{equation}\label{CouplingDensityFormula}
\begin{split}
P_{n,m}(\underline{y})=&\frac{1}{Z_{n,m}}\det\left[\left(y_j^m\right)^{\frac{k-1}{2}}I_{k-1}\left(2b\left(y_j^m\right)^\frac{1}{2}\right)e^{-V_m\left(y_j^m\right)}\right]_{j,k=1}^n\\
&\times
\prod\limits_{l=1}^{m-1}\det\left[\frac{\left(y_j^{l+1}\right)^{\nu_{l+1}}}{\left(y_k^l\right)^{\nu_{l+1}+1}}e^{-\frac{y_j^{l+1}}{y_k^l}
-V_{l}\left(y_k^{l}\right)}\right]_{j,k=1}^n
\det\left[(y_j^1)^{\nu_1+k-1}e^{-y_j^1}\right]_{j,k=1}^n,
\end{split}
\end{equation}
where
\begin{equation}
\label{Znmgen}
Z_{n,m}=\left(n!\right)^m\det\left[a_{i,j}\right]_{i,j=1}^n,
\end{equation}
the matrix elements $a_{i,j}$ are given by
\begin{equation}\label{GeneralMatrixElements}
\begin{split}
&a_{i,j}=\int\limits_0^{\infty}\ldots\int\limits_0^{\infty}
t_1^{\nu_1+i-1}\left(\frac{t_2}{t_1}\right)^{\nu_2}\ldots\left(\frac{t_{m-1}}{t_{m-2}}\right)^{\nu_{m-1}}
\left(\frac{t_m}{t_{m-1}}\right)^{\nu_m}
\left(t_m\right)^{\frac{j-1}{2}}I_{j-1}\left(2b\left(t_m\right)^{\frac{1}{2}}\right)
\\
&\times e^{-t_1-\frac{t_2}{t_1}-\ldots-\frac{t_{m-1}}{t_{m-2}}-\frac{t_m}{t_{m-1}}
-V_1\left(t_1\right)-V_2\left(t_2\right)-\ldots-V_{m-1}\left(t_{m-1}\right)-V_{m}\left(t_{m}\right)}
\frac{dt_1}{t_1}\frac{dt_2}{t_2}\ldots \frac{dt_{m-1}}{t_{m-1}} dt_m,
\end{split}
\end{equation}
and where
$I_{\mu}(z)$ denotes the modified Bessel function of the first kind.
\end{thm}
Recall that the modified Bessel function of the first kind $I_{\mu}(z)$ is defined by
\begin{equation}
I_{\mu}(z)=\sum\limits_{m=0}^{\infty}\frac{1}{m!\Gamma(\mu+m+1)}\left(\frac{z}{2}\right)^{2m+\mu},
\label{I-def}
\end{equation}
and that the function $y^{\frac{\mu}{2}}I_{\mu}\left(2by^{\frac{1}{2}}\right)$ has the following asymptotic behaviour
\begin{equation}\label{ISmallb}
y^{\frac{\mu}{2}}I_{\mu}\left(2by^{\frac{1}{2}}\right)\sim\frac{\left(by\right)^{\mu}}{\Gamma(\mu+1)},\;\; b\rightarrow 0.
\end{equation}
From (\ref{CouplingDensityFormula}) and (\ref{ISmallb})
we conclude that as $V_1=\ldots=V_m=0$, and as $b$ approaches  zero, the joint probability density $P_{n,m}(\underline{y})$ becomes equal to
\begin{equation}\label{GinibreDensityFormula}
\frac{1}{Z_{n,m}^{\Ginibre}}
\Delta_n(\{y_j^m\})
\prod\limits_{l=1}^{m-1}
\det\left[\frac{\left(y_j^{l+1}\right)^{\nu_{l+1}}}{\left(y_k^l\right)^{\nu_{l+1}+1}}e^{-\frac{y_j^{l+1}}{y_k^l}}\right]_{j,k=1}^n
\det\left[(y_j^1)^{\nu_1+k-1}e^{-y_j^1}\right]_{j,k=1}^n,
\end{equation}
where
\begin{equation}
\label{Znm}
Z_{n,m}^{\Ginibre}=\left(n!\right)^m\prod\limits_{j=1}^n\prod\limits_{l=0}^m\Gamma\left(j+\nu_l\right),
\end{equation}
and we have defined the Vandermonde determinant
\begin{equation}
\label{Vandermonde}
\Delta_n(\{ y_j^l\})=\det\left[\left(y_j^l\right)^{k-1}\right]_{j,k=1}^n=\prod_{1\leq j<k\leq n}(y_k^l-y_j^l)\ .
\end{equation}
This is the density for the Ginibre product process, cf. Strahov \cite[Proposition 4.2]{StrahovD} (see also  \cite{AkemannIpsenKieburg}).
Note that the last determinant in \eqref{GinibreDensityFormula} is proportional to the Vandermonde determinant $\Delta_n(\{y_j^1\})$.

The vector $y^m=\left(y_1^m,\ldots,y^m_n\right)$ is that of the squared singular values of the total product matrix
$Y_m=G_m\cdots G_1$. In view of  the recent results on distributions of squared singular values for products of random matrices
mentioned in the introduction, the probability distribution of $y^m$ is of special interest. Corollary \ref{Corrolary1} from Theorem \ref{TheoremGinibreCouplingDensity}
gives the joint probability distribution of $y^m$ explicitly.
\begin{cor}\label{Corrolary1} Consider the multi-matrix model defined by the probability distribution (\ref{MainProbabilityMeasure}), and form the total product
$Y_m=G_m\cdots G_1$.
The squared singular values of $Y_m$ form a biorthogonal ensemble with the joint probability density function given by
$$
P_{n,m}\left(y_1^m,\ldots,y_n^m\right)=
\frac{1}{Z_{n,m}'}\det\left[\varphi_i\left(y_j^m\right)\right]_{i,j=1}^n
\det\left[\psi_i\left(y_j^m\right)\right]_{i,j=1}^n,
$$
with the normalisation $Z_{n,m}'=n!\det\left[a_{i,j}\right]_{i,j=1}^n$ resulting from \eqref{Znmgen}. The functions $\varphi_i$, with $i=1,\ldots, n$
are given by
$$
\varphi_i\left(y\right)=y^{\frac{i-1}{2}}I_{i-1}\left(2by^{\frac{1}{2}}\right)e^{-V_m(y)},
$$
and the functions $\psi_i$
with $=1,\ldots,n$ are given by
\begin{equation}
\begin{split}
\psi_i(y)=&\int\limits_0^{\infty}\ldots\int\limits_0^{\infty}
t_1^{\nu_1+i-1}\left(\frac{t_2}{t_1}\right)^{\nu_2}\ldots\left(\frac{t_{m-1}}{t_{m-2}}\right)^{\nu_{m-1}}
\left(\frac{y}{t_{m-1}}\right)^{\nu_m}
\\
&\times e^{-t_1-\frac{t_2}{t_1}-\ldots-\frac{t_{m-1}}{t_{m-2}}-\frac{y}{t_{m-1}}
-V_1\left(t_1\right)-V_2\left(t_2\right)-\ldots-V_{m-1}\left(t_{m-1}\right)}
\frac{dt_1}{t_1}\frac{dt_2}{t_2}\ldots \frac{dt_{m-1}}{t_{m-1}}.
\end{split}
\end{equation}
\end{cor}
For some particular choices of potentials $V_1$, $\ldots$, $V_m$ we can obtain explicit expressions
for the functions $\psi_i(y)$, and for the normalisation constant $Z_{n,m}'$, as we will show below.

First, let us suppose that we could set all potentials
$V_1=\ldots=V_m=0$. Then the functions $\psi_i(y)$ could be written as Meijer $G$-functions\footnote{We
refer the reader to the book by Luke \cite{Luke} for the definition of Meijer $G$-functions, and for their exact and asymptotic properties.},
$$
\psi_i(y)=G^{m,0}_{0,m}\left(\begin{array}{cccc}
                                & - &  &  \\
                               \nu_1+i-1, & \nu_2, & \ldots, & \nu_m
                             \end{array}
\biggl|y\right).
$$
For the resulting normalisation we would obtain from Andr\'eief's integral identity
\begin{equation}\label{NormalizationConstant1}
Z_{n,m}'=n!\det\left[\int\limits_0^{\infty}y^{\frac{i-1}{2}}
I_{i-1}\left(2by^{\frac{1}{2}}\right)
G^{m,0}_{0,m}\left(\begin{array}{cccc}
                                & - &  &  \\
                               \nu_1+j-1, & \nu_2, & \ldots, & \nu_m
                             \end{array}
\biggl|y\right)
dy\right]_{i,j=1}^{n}.
\end{equation}
However, for $m>2$ the integrals inside the determinant do not converge. The reason is that the modified Bessel function of the second kind \eqref{I-def} has the following asymptotic
\begin{equation}\label{A1}
I_{i-1}\left(2by^{\frac{1}{2}}\right)\sim\frac{e^{2by^{\frac{1}{2}}}}{2\pi^{\frac{1}{2}}b^{\frac{1}{2}}y^{\frac{1}{4}}},\;\; y\rightarrow\infty,
\end{equation}
c.f. \cite{NIST},
which has to be compared to the asymptotic of the Meijer $G$-function \cite{Luke}
\begin{equation}
G^{m,0}_{0,m}\left(\begin{array}{cccc}
                                & - &  &  \\
                               \nu_1+j-1, & \nu_2, & \ldots, & \nu_m
                             \end{array}
\biggl|y\right)\sim
\frac{(2\pi)^{\frac{m-1}{2}}}{m^{\frac12}} y^\theta \exp[-my^{\frac{1}{m}}]
,\;\; y\rightarrow\infty,
\label{Gasympt}
\end{equation}
with $\theta=\frac{1}{m}(\frac12(1-m)+\nu_1+j-1+\sum_{l=2}^m\nu_l)$. Only for $m=2$ the two exponentials in eqs. \eqref{A1} and \eqref{Gasympt}
together lead to convergent integrals for $b<1$,
and we obtain the biorthogonal ensemble equivalent to that studied by the authors in \cite{AkemannStrahov,AkemannStrahov1}\footnote{For $m=1$ the corresponding one-matrix model is convergent for any $b$.}.
Indeed for $m=2$ we  have (recall that $\nu_2=0$ then)
$$
G^{2,0}_{0,2}\left(\begin{array}{cc}
                                & -   \\
                               \nu_1+j-1, & 0
                             \end{array}
\biggl|y\right)=2y^{\frac{\nu_1+j-1}{2}}K_{\nu_1+j-1}\left(2y^{\frac{1}{2}}\right),
$$
where $K_{\kappa}(z)$ denotes the modified Bessel function of the second kind.
It can be defined by the integral
\begin{equation}
\label{K-def1}
K_{\kappa}(z)=\frac{\Gamma\left(\kappa+\frac{1}{2}\right)(2z)^{\kappa}}{\sqrt{\pi}}
\int\limits_0^{\infty}\frac{\cos(t)dt}{\left(t^2+z^2\right)^{\kappa+\frac{1}{2}}}\,,\;\;\Re(\kappa)>-\frac12\,,
\end{equation}
which is more convenient for complex contour integrals over $\kappa$, or alternatively \cite{NIST} as
\begin{equation}
\label{K-def2}
K_\kappa(2z)=\frac12 z^\kappa \int_0^\infty dt\,t^{-\kappa-1} \exp\left[-t-\frac{z^2}{t}\right]\, .
\end{equation}
From the latter is is easy to obtain the following asymptotic expansion \cite[10.25.3]{NIST}, corresponding to \eqref{Gasympt} for $m=2$:
\begin{equation}\label{A2}
K_{\kappa}\left(2by^{\frac{1}{2}}\right)\sim\frac{\pi^{\frac{1}{2}}}{2^{\frac{1}{2}}\left(2by^{\frac{1}{2}}\right)^{\frac{1}{2}}}e^{-2by^{\frac{1}{2}}},
\;\; y\rightarrow\infty.
\end{equation}
It is not difficult to see that the model for $m=2$ studied in \cite{AkemannStrahov} in the parametrisation given in \eqref{m=2-model} can be mapped to our parametrisation \eqref{MainProbabilityMeasure}, with $a=1$ and $b>0$ by rescaling in \eqref{m=2-model} $G_1\to\alpha_1G_1$ and $G_2\to\alpha_2 G_2$, with $\alpha_1^2=\frac{2}{1+\mu}$ and $\alpha_2^2=\frac{1+\mu}{2\mu}$, leading to $b=\frac{1-\mu}{\sqrt{4\mu}}$. Therefore, in order to  guarantee the existence  of the corresponding matrix measure (\ref{MainProbabilityMeasure}),
for $m>2$ some non-zero potentials $V_l$ should be added for non-zero coupling $b>0$, to ensure convergence of the integrals.

In this paper the case corresponding to
\begin{equation}
\label{V-cond}
V_1(t)=\ldots=V_{m-2}(t)=0\ , \ \ 
V_{m-1}(t)=b^2t\ , \ \ V_m(t)=0\ ,
\end{equation}
will be considered in detail\footnote{Note that for $m=2$ the additional potential 
$V_1(t)=bt^2$ simply corresponds to a shift of the Gaussian term $\Tr [G_1^*G_1]$ which is already present.}. Here, we are dealing with coupled matrices with the single coupling 
constant
$b>0$. In this case, a simple contour integral representation for the correlation kernel suitable for an asymptotic analysis can be derived. Inserting the above conditions on the potentials \eqref{V-cond} into
the matrix measure (\ref{MainProbabilityMeasure}), it turns into
\begin{equation}\label{MainProbabilityMeasure1}
 \begin{split}
&\frac{1}{Z_n(b)}e^{-\sum\limits_{l=1}^m\Tr\left[G_l^*G_l\right]+b\Tr\left[G_m\cdots G_1+\left(G_m\cdots G_1\right)^*\right]
-b^2\Tr\left[\left(G_{m-1}\cdots G_1\right)^*\left(G_{m-1}\cdots G_1\right)\right]}\prod\limits_{l=1}^mdG_l\\
=&
\frac{1}{Z_n(b)}e^{-\sum\limits_{l=1}^{m-1}\Tr\left[G_l^*G_l\right]
-\Tr\left[\left(G_m-b(G_{m-1}\cdots G_1)^*\right)
\left(G_m^*-bG_{m-1}\cdots G_1\right)
\right]}\prod\limits_{l=1}^mdG_l\,.
\end{split}
\end{equation}
The normalised probability density of $\underline{y}=\left(y^m,\ldots,y^1\right)$ that follows from Theorem \ref{TheoremGinibreCouplingDensity} by inserting \eqref{V-cond} into \eqref{CouplingDensityFormula} is given by
\begin{equation}\label{CouplingDensityFormula1}
\begin{split}
P_{n,m}(\underline{y};b)=
&\frac{1}{Z_{n,m}(b)}\det\left[\left(y_j^m\right)^{\frac{k-1}{2}}I_{k-1}\left(2b\left(y_j^m\right)^\frac{1}{2}\right)\right]_{j,k=1}^n
\det\left[\frac{1}{y_k^{m-1}}e^{-\frac{y_j^{m}}{y_k^{m-1}}-b^2y_k^{m-1}}\right]_{j,k=1}^n\\
&\times
\prod\limits_{l=1}^{m-2}\det\left[\frac{\left(y_j^{l+1}\right)^{\nu_{l+1}}}{\left(y_k^l\right)^{\nu_{l+1}+1}}e^{-\frac{y_j^{l+1}}{y_k^l}}\right]_{j,k=1}^n
\det\left[(y_j^1)^{\nu_1+k-1}e^{-y_j^1}\right]_{j,k=1}^n,
\end{split}
\end{equation}
where the normalisation constant $Z_{n,m}(b)$ following from \eqref{Znmgen} remains to be determined, see \eqref{Znmb} below.
In what follows we will refer to the point process formed by random  configuration $\underline{y}$
with the joint probability density given by equation (\ref{CouplingDensityFormula1}) as the \textit{Ginibre product process with coupling}.
The Ginibre product process with coupling is an one-parameter \textit{deformation} of the Ginibre product process, and the coupling constant $b$ is the parameter of this deformation.
\begin{rem}
The model defined by equation \eqref{MainProbabilityMeasure1} has the following interpretation. According to equation \eqref{MainProbabilityMeasure1} the matrices
$G_1$, $\ldots$, $G_m$ and $G_m'=G_m-b\left(G_{m-1}\ldots G_1\right)^*$ can be understood as independent Ginibre matrices (each of the size
$\left(n+\nu_l\right)\times\left(n+\nu_{l-1}\right)$, where $l=1,\ldots,m$). Therefore, we will have the Ginibre product process on the first $m-1$ levels which will be independent 
of $b$. As $b\rightarrow 0$, the matrix $G_m-b\left(G_{m-1}\ldots G_1\right)^*$ turns into the Ginibre matrix $G_m$ independent 
of $G_1$, $\ldots$, $G_{m-1}$, and we obtain the Ginibre product process on all $m$ levels. As $b\rightarrow\infty$, the matrix $G_m'$ is dominated by 
$\left(G_{m-1}\ldots G_1\right)^*$, so the singular values of $G_m'G_{m-1}\ldots G_1$ will approach the squares of the singular values of $G_{m-1}\ldots G_1$.
\end{rem}
\begin{rem}
If in the second line of \eqref{MainProbabilityMeasure1} we rescale $G_m\to b G_m$ and then send $b\to\infty$ this leads to a delta function constraint,
due to $\lim_{b\to\infty}b\exp[-b^2(x-y)^2]/\sqrt{\pi}=\delta(x-y)$,
enforcing
$G_m=(G_{m-1}\cdots G_1)^*$ at finite matrix sizes. It is precisely this limit that leads to the collapse of the Ginibre product process with coupling having $m$ levels to a Ginibre product process having only $m-1$ levels, as we will see below on the level of the joint probability density \eqref{CouplingDensityFormula1}
for finite $n$, and in the large-$n$ limit in Theorem \ref{TheoremHardEdgeGinibreCouplingProcess} (C).
\end{rem}

It is instructive to consider separately the distribution of the squared singular values of the total product matrix
$Y_m=G_m\cdots G_1$, in the case when the joint distribution of $G_m$, $\ldots$, $G_1$ is given by
formula (\ref{MainProbabilityMeasure1}), with joint probability density \eqref{CouplingDensityFormula1}. We obtain the following result.
\begin{prop}\label{PropositionJointDensityTotalProductMatrix}
Consider the probability measure defined by equation (\ref{MainProbabilityMeasure1}), and
let $Y_m=G_m\cdots G_1$ be the total product matrix of the corresponding multi-matrix model.
The squared singular values $y_1^m,\ldots,y_n^m$  of $Y_m$ form a determinantal point process on $\R_{>0}$.
This determinantal point process is a biorthogonal ensemble defined by the joint probability density
$P\left(y_1^m,\ldots,y_n^m;b\right)$. 
The explicit formula for $P\left(y_1^m,\ldots,y_n^m;b\right)$ reads
\begin{equation}\label{PExact}
P\left(y_1^m,\ldots,y_n^m;b\right)=\frac{(n!)^{m-1}}{Z_{n,m}(b)}\det\left[\left(y_j^m\right)^{\frac{k-1}{2}}I_{k-1}\left(2b\left(y_j^m\right)^{\frac{1}{2}}\right)\right]_{k,j=1}^n
\det\left[\psi_k(y_j^m)\right]_{k,j=1}^n,
\end{equation}
where
\begin{equation}\label{PSIK}
\begin{split}
&\psi_k\left(y\right)=\frac{1}{2\pi i}\int\limits_{c-i\infty}^{c+i\infty}\prod\limits_{l=2}^m\Gamma\left(u+\nu_l\right)
\Gamma\left(u+\nu_1+k-1\right)\frac{2\left(by^{\frac{1}{2}}\right)^uK_u\left(2by^{\frac{1}{2}}\right)}{\Gamma(u)}y^{-u}du,
\end{split}
\end{equation}
with $c>0$ such that the poles of the Gamma-functions are to the right of the contour. The normalisation constant is
given by
\begin{equation}
\label{Znmb}
Z_{n,m}(b)=\left(n!\right)^mb^{\frac{n(n-1)}{2}}\prod\limits_{j=1}^n\prod\limits_{l=1}^{m-1}\Gamma\left(j+\nu_l\right).
\end{equation}
\end{prop}
Note that following \cite{NIST} we have
\begin{equation}
\label{KSmallb}
\frac{2\left(by^{\frac{1}{2}}\right)^uK_u\left(2by^{\frac{1}{2}}\right)}{\Gamma(u)}\sim1,\;\; b\rightarrow 0,\;\; \mbox{for}\;\; \Re(u)>0.
\end{equation}
Taking this into account, and using the asymptotic expression 
$y^{\frac{j-1}{2}}I_{j-1}\left(2by^{\frac{1}{2}}\right)$
(see equation (\ref{ISmallb}), contributing with an extra factor $1/\Gamma(j)$), we find that the joint probability density function
$P\left(y_1,\ldots,y_n;b\right)$  defined by equation  (\ref{PExact}) has a limit as $b\rightarrow 0$.
Namely, we have
\begin{equation}
\begin{split}
\underset{b\rightarrow 0}{\lim}&P\left(y_1^m,\ldots,y_n^m;b\right)=
\frac{1}{n!\prod\limits_{i=1}^n\prod\limits_{j=0}^m\Gamma\left(j+\nu_j\right)}\\
&\times\Delta_n(\{y_j^m\})
\det\left[G^{m,0}_{0,m}\left(
\begin{array}{cccc}
   & - &  &  \\
  \nu_1+k-1, & \nu_2, & \ldots, & \nu_m
\end{array}
\biggl|y_j^m\right)\right]_{j,k=1}^n.
\end{split}
\end{equation}
This is the joint probability density function obtained by Akemann,  Ipsen, and Kieburg \cite{AkemannIpsenKieburg} for the squared singular values of the ensemble of the products of $m$ independent rectangular matrices with independent complex Gaussian entries.

Now let us consider the behaviour of $P\left(y_1^m,\ldots,y_n^m;b\right)$ 
as $b\rightarrow\infty$. Using the asymptotic formulae \eqref{A1} and \eqref{A2} of the modified Bessel functions involved,  we find that as $b\rightarrow\infty$, the joint probability density function
$P\left(y_1^m,\ldots,y_n^m;b\right)$  defined by equation  (\ref{PExact}) becomes approximately equal to
\begin{equation}
\begin{split}
&\sim\frac{1}{2^nb^{\frac{n(n+1)}{2}}n!
\prod\limits_{i=1}^n\prod\limits_{j=0}^{m-1}\Gamma\left(j+\nu_j\right)
\prod\limits_{j=1}^ny_j^{\frac{1}{2}}}\\
&\times\Delta_n(\{y_j^m\})
\det\left[
G^{m-1,0}_{0,m-1}\left(
\begin{array}{cccc}
   & - &  &  \\
  \nu_1+k-1, & \nu_2, & \ldots, & \nu_{m-1}
\end{array}
\biggl|\frac{1}{b}{\left(y_j^m\right)^{\frac{1}{2}}}\right)\right]_{k,j=1}^n.
\end{split}
\nonumber
\end{equation}
Set $x_j=\frac{1}{b}(y_j^m)^{\frac{1}{2}}$. We then find that as $b\rightarrow\infty$, the joint probability density function of the new variables $x_1$, $\ldots$, $x_n$ converges to
\begin{equation}
\frac{1}{n!\prod\limits_{i=1}^n\prod\limits_{j=0}^{m-1}\Gamma\left(j+\nu_j\right)}\\
\Delta_n(\{x_j\})
\det\left[G^{m-1,0}_{0,m-1}\left(
\begin{array}{cccc}
   & - &  &  \\
  \nu_1+k-1, & \nu_2, & \ldots, & \nu_{m-1}
\end{array}
\biggl|x_j\right)\right]_{j,k=1}^n.
\end{equation}
This is the joint probability density function for the ensemble of the squared singular values of the product of $m-1$ rectangular
matrices with independent complex Gaussian entries. We conclude that the biorthogonal ensemble defined by equation (\ref{PExact})
is an interpolating ensemble: it interpolates between the process of  squared singular values from the ensemble of the products of $m$ independent rectangular
matrices with independent complex Gaussian entries, and that of $m-1$ independent rectangular
matrices with independent complex Gaussian entries.

As a final example we give the joint density of the product matrix process related to Hurwitz numbers via \cite{AmbjornChekov}, as mentioned in the introduction. Compared to
\eqref{MainProbabilityMeasure} it reads
 \begin{equation}\label{jpdfHurwitz}
\frac{1}{Z_n}
\exp\left[-\Tr\left[QG_1^*G_1\right] -\sum\limits_{l=2}^m\Tr\left[G_l^*G_l\right]
-\Tr \left[V_m\left(\left(G_{m}\cdots G_1\right)^*G_{m}\cdots G_1\right)\right]\right]
\prod\limits_{l=1}^mdG_l,
\end{equation}
that is $V_1=\ldots=V_{m-1}=0$ vanish, and in addition the Gaussian distribution of matrix $G_1$ now includes a nontrivial covariance matrix $Q$ of size $n\times n$ with positive
eigenvalues $q_1,\ldots q_n>0$. 
\begin{prop}
Consider a multi-matrix model defined by probability measure \eqref{jpdfHurwitz}. Denote by $\left\{y_j^l\right\}_{j=1,\ldots,m}$ the set of the squared singular values of the matrix $Y_l=G_l\ldots G_1$ for each $l=1,\ldots,m$. The joint probability density of $\underline{y}=\left(y^m,\ldots,y^1\right)$ is proportional to
\end{prop}
\begin{equation}\label{HurwitzDensityFormula}
\det\left[\left(y_j^m\right)^{k-1}e^{-V_m(y_j^m)}\right]_{j,k=1}^n
\prod\limits_{l=1}^{m-1}
\det\left[\frac{\left(y_j^{l+1}\right)^{\nu_{l+1}}}{\left(y_k^l\right)^{\nu_{l+1}+1}}e^{-\frac{y_j^{l+1}}{y_k^l}}\right]_{j,k=1}^n
\det\left[(y_j^1)^{\nu_1}e^{-y_j^1q_k}\right]_{j,k=1}^n.
\end{equation}
The joint density of $\underline{y}$ for the model defined by \eqref{jpdfHurwitz} can be obtained similarly to the proof of Theorem 2.1,  
using the standard Harish-Chandra-Itzykson-Zuber formula. In the following we will not pursue this example further.

\subsection{Exact formulae for the correlation kernels}
Theorem \ref{TheoremGinibreCouplingDensity} states that the density of the product matrix process associated
with the multi-matrix model (\ref{MainProbabilityMeasure}) can be written as a product of determinants.
This enables us to apply the result by  Eynard and Mehta \cite{EynardMehta}, and to give  a formula for all correlation functions. Namely,
the Eynard-Mehta Theorem implies that the configuration of the squared singular values of all product matrices associated with the multi-matrix model
(\ref{MainProbabilityMeasure}) is  a \textit{determinantal point process} on
 $\left\{1,\ldots,m\right\}\times\R_{>0}$. The correlation kernel of this determinantal point process is given by the next theorem.
\begin{thm}\label{TheoremMostGerneralCorrelationKernel}
Consider the multi-matrix model defined by the probability measure (\ref{MainProbabilityMeasure}), with $a=1$.
Denote by the set $\{y_1^l,\ldots,y_n^l\}$ the squared singular values of the matrix
$Y_l=G_l\cdots G_1$, that are the eigenvalues of the matrix $Y_l^*Y_l$.
The configuration of all these eigenvalues
\begin{equation}
\left\{\left(l,y_j^l\right)\vert l=1,\ldots,m;j=1,\ldots,n\right\}
\nonumber
\end{equation}
forms  a determinantal point process
on $\left\{1,\ldots,m\right\}\times\R_{>0}$.  The correlation kernel of this determinantal point process, $K_{n,m}^{V}(r,x;s,y)$,
where $r,s\in\left\{1,\ldots,m\right\}$, and $x, y\in \R_{>0}$, can be written as
\begin{equation}
\label{KnmV}
K_{n,m}^{V}(r,x;s,y)=-\phi_{r,s}(x,y)+\sum\limits_{i,j=1}^n\phi_{r,m+1}(x,i)\left(A^{-1}\right)_{i,j}\phi_{0,s}(j,y).
\end{equation}
Here, the elements of matrix $A=(a_{i,j})$, with $i,j=1,\ldots,n$, are defined in \eqref{GeneralMatrixElements}.
The kernel depends on three sets of functions which are given as follows.

\textit{(i)}
For the first set of functions $\phi_{r,s}(x,y)$ we distinguish the following cases:
\begin{itemize}
\item for $r=1,\ldots,m-1$ we have:
\begin{equation}
\label{PHI1rr+1}
\phi_{r,r+1}(x,y)=\left(\frac{y}{x}\right)^{\nu_{r+1}}\frac{e^{-\frac{y}{x}-V_{r}(x)}}{x}.
\end{equation}

\item for $1\leq r\leq m-2$ and $r+2\leq s\leq m$ we have:
\begin{equation}
\label{PHI1rs}
\begin{split}
\phi_{r,s}(x,y)=&\frac{e^{-V_r(x)}}{x}\int\limits_0^{\infty}\ldots\int\limits_0^{\infty}
\left(\frac{t_{r+1}}{x}\right)^{\nu_{r+1}}
\left(\frac{t_{r+2}}{t_{r+1}}\right)^{\nu_{r+2}}
\ldots
\left(\frac{y}{t_{s-1}}\right)^{\nu_{s}}
\\
&\times e^{-\frac{t_{r+1}}{x}-\frac{t_{r+2}}{t_{r+1}}-\ldots-\frac{t_{s-1}}{t_{s-2}}-\frac{y}{t_{s-1}}-V_{r+1}\left(t_{r+1}\right)
-V_{r+2}\left(t_{r+2}\right)\ldots-V_{s-1}\left(t_{s-1}\right)}
\\
&\times
\frac{dt_{r+1}}{t_{r+1}}\frac{dt_{r+2}}{t_{r+2}}\ldots \frac{dt_{s-2}}{t_{s-2}}\frac{dt_{s-1}}{t_{s-1}}.
\end{split}
\end{equation}

\item for $1\leq s\leq r\leq m$ we have:
\begin{equation}
\label{PHI1s<=r}
\phi_{r,s}(x,y)=0.
\end{equation}

\end{itemize}

\textit{(ii)}
For the second set of functions $\phi_{r,m+1}(x,j)$ we distinguish the following cases:
\begin{itemize}
\item for $1\leq r\leq m-1$  we have :
\begin{equation}
\label{PHI2rm+1}
\begin{split}
\phi_{r,m+1}(x,j)=&\frac{e^{-V_r(x)}}{x}\int\limits_0^{\infty}\ldots\int\limits_0^{\infty}
\left(\frac{t_{r+1}}{x}\right)^{\nu_{r+1}}
\left(\frac{t_{r+2}}{t_{r+1}}\right)^{\nu_{r+2}}
\ldots
\left(\frac{t_m}{t_{m-1}}\right)^{\nu_{m}}
t_m^{\frac{j-1}{2}}I_{j-1}\left(2bt_m^{\frac{1}{2}}\right)
\\
&\times e^{-\frac{t_{r+1}}{x}-\frac{t_{r+2}}{t_{r+1}}-\ldots-\frac{t_{m-1}}{t_{m-2}}-\frac{t_m}{t_{m-1}}-V_{r+1}\left(t_{r+1}\right)
-V_{r+2}\left(t_{r+2}\right)\ldots-V_{m-1}\left(t_{m-1}\right)-V_{m}\left(t_{m}\right)}\\
&\times\frac{dt_{r+1}}{t_{r+1}}\frac{dt_{r+2}}{t_{r+2}}\ldots \frac{dt_{m-1}}{t_{m-1}}dt_m.
\end{split}
\end{equation}

\item for $r=m$ we have:
\begin{equation}
\label{PHI2mm+1}
\phi_{m,m+1}(x,j)=x^{\frac{j-1}{2}}I_{j-1}\left(2bx^{\frac{1}{2}}\right)e^{-V_m(x)}.
\end{equation}
\end{itemize}

\textit{(iii)}
For the third set of functions $\phi_{0,s}(i,y)$ we distinguish the following cases:
\begin{itemize}
\item for $2\leq s\leq m$ we have:
\begin{equation}
\label{PHI30s}
\begin{split}
\phi_{0,s}(i,y)=&\int\limits_0^{\infty}\ldots\int\limits_0^{\infty}
t_1^{\nu_1+i-1}\left(\frac{t_2}{t_1}\right)^{\nu_2}\ldots\left(\frac{t_{s-1}}{t_{s-2}}\right)^{\nu_{s-1}}
\left(\frac{y}{t_{s-1}}\right)^{\nu_s}
\\
&\times e^{-t_1-\frac{t_2}{t_1}-\ldots-\frac{t_{s-1}}{t_{s-2}}-\frac{y}{t_{s-1}}-
V_1\left(t_1\right)-V_2\left(t_2\right)-\ldots-V_{s-2}\left(t_{s-2}\right)-V_{s-1}\left(t_{s-1}\right)}\\
&\times
\frac{dt_1}{t_1}\frac{dt_2}{t_2}\ldots \frac{dt_{s-2}}{t_{s-2}}\frac{dt_{s-1}}{t_{s-1}}.
\end{split}
\end{equation}

\item for $s=1$ we have:
\begin{equation}
\label{PHI301}
\phi_{0,1}(i,y)=y^{\nu_1+i-1}e^{-y}.
\end{equation}
\end{itemize}

\end{thm}
Recall that  the case corresponding to $V_1(t)=\ldots=V_{m-2}(t)=0$, 
$V_{m-1}(t)=b^2t$, and $V_m(t)=0$
corresponds to the Ginibre product process with coupling. Let us denote the relevant correlation kernel by
$K_{n,m}(r,x;s,y;b)$, to emphasise its dependence on the coupling constant $b$.
 For the Ginibre product process with coupling
 we are able to reduce all multiple integrals in the statement of Theorem \ref{TheoremMostGerneralCorrelationKernel}
 to at most single integrals, expressed
 in terms of special functions, and to find the matrix $A^{-1}$ explicitly. This leads to a contour integral
 representation for the correlation kernel  $K_{n,m}(r,x;s,y;b)$ as follows.
 \begin{thm}\label{TheoremDoubleIntegralRepresentationExactKernel}
For the special case $V_1(t)=\ldots= V_{m-2}(t)=0$, 
$V_{m-1}(t)=b^2t$, and $V_m(t)=0$ the correlation kernel
of Theorem \ref{TheoremMostGerneralCorrelationKernel} can be written  as
\begin{equation}\label{CorrelationKernelFormulaContour}
\begin{split}
&K_{n,m}(r,x;s,y;b)=-\phi_{r,s}(x,y;b)+S_{n,m}(r,x;s,y;b).
\end{split}
\end{equation}
The functions $\phi_{r,s}(x,y;b)$ are given by
\begin{equation}\label{PHI(x,y,b)}
\begin{split}
&\phi_{r,s}(x,y;b)\\
&=\left\{
                  \begin{array}{llll}
                  \frac{e^{-\frac{y}{x}-b^2x}}{x}, & r=m-1, s=m,\\
                    \frac{1}{x}
G^{s-r,0}_{0,s-r}\left(\begin{array}{cccc}
                - \\
               \nu_{r+1},\ldots,\nu_{s}
             \end{array}\biggl|\frac{y}{x}\right), &  1\leq r<s\leq m-1,\\
                    \frac{1}{x}
\int\limits_{0}^{\infty}
G^{m-r-1,0}_{0,m-r-1}\left(\begin{array}{cccc}
                - \\
               \nu_{r+1},\ldots,\nu_{m-1}
             \end{array}\biggl|\frac{t}{x}\right)
             e^{-\frac{y}{t}-b^2t}\frac{dt}{t}, & 1\leq r\leq m-2, s=m,\\
0, & 1\leq s\leq r\leq m.
                  \end{array}
                \right.
\end{split}
\end{equation}
The functions $S_{n,m}(r,x;s,y;b)$ can be written as
\begin{equation}
\label{Snm}
\begin{split}
S_{n,m}(r,x;s,y;b)=&\frac{1}{(2\pi i)^2}\int\limits_{-\frac{1}{2}-i\infty}^{-\frac{1}{2}+i\infty}du
\oint\limits_{\Sigma_n}dt\frac{\Gamma(t-n+1)}{\Gamma(u-n+1)}
\frac{\prod_{j=0}^s\Gamma(u+\nu_j+1)}{\prod_{j=0}^r\Gamma(t+\nu_j+1)}\\
&\times\frac{x^tp_r\left(t,x;b\right)y^{-u-1}q_s\left(u+1,y;b\right)}{u-t}.
\end{split}
\end{equation}
Here,  $\Sigma_n$ is a closed contour encircling $0,1,\ldots,n$ in positive direction
and such that $\re t>-\frac{1}{2}$ for $t\in\Sigma_n$.
In \eqref{Snm}
the functions $p_r(t,x;b)$ are defined by
\begin{equation}\label{Functionp}
p_r(t,x;b)=\left\{
             \begin{array}{ll}
               1, & r\in\left\{1,\ldots,m-1\right\}, \\
               \frac{\Gamma(t+1)I_t\left(2bx^{\frac{1}{2}}\right)}{\left(bx^{\frac{1}{2}}\right)^t}, & r=m,
             \end{array}
           \right.
\end{equation}
and the functions  $q_s\left(u,y;b\right)$ are defined by
\begin{equation}\label{Functionq}
q_s\left(u,y;b\right)=\left\{
             \begin{array}{ll}
               1, & s\in\left\{1,\ldots,m-1\right\}, \\
               \frac{2\left(by^{\frac{1}{2}}\right)^uK_u\left(2by^{\frac{1}{2}}\right)}{\Gamma(u)}, & s=m.
             \end{array}
           \right.
\end{equation}
\end{thm}
\begin{Remarks}
(a) As a determinantal point process between the first $m-1$ levels, that is on
 $\left\{1,\ldots,m-1\right\}\times\R_{>0}$, the kernel  \eqref{CorrelationKernelFormulaContour} is independent of $b$ and agrees with
kernel $K_{n,m-1}^{\Ginibre}(r,x;s,y)$ of the Ginibre product process  \textit{without coupling} found by Strahov \cite[Prop.2.3]{StrahovD}:
\begin{equation}\label{K1}
\begin{split}
K_{n,m-1}^{\Ginibre}(r,x;s,y)=&-\frac{1}{x}G^{s-r,0}_{0,s-r}\left(\begin{array}{ccc}
                - \\
               \nu_{r+1},\ldots,\nu_s
             \end{array}\biggl|\frac{y}{x}\right)\mathbf{1}_{s>r}\\
             &+\frac{1}{(2\pi i)^2}\int\limits_{-\frac{1}{2}-i\infty}^{-\frac{1}{2}+i\infty}
du\oint\limits_{\Sigma_n}dt\frac{\prod_{j=0}^s\Gamma(u+\nu_j+1)}{\prod_{j=0}^r\Gamma(t+\nu_j+1)}
\frac{\Gamma(t-n+1)}{\Gamma(u-n+1)}
\frac{x^ty^{-u-1}}{u-t},
\end{split}
\end{equation}
where $\mathbf{1}_{s>r}$ is the indicator function\footnote{We follow the convention here that the index $m-1$ of the kernel on the left-hand side indicates the range that the arguments $r$ and $s$ can take.}. Because this agreement holds at finite matrix size $n$ it will also hold for the limiting kernel. For $r=m$ or $s=m$  the kernel $K_{n,m}(r,x;s,y;b)$ \eqref{CorrelationKernelFormulaContour} depends on $b$ and differs from the kernel \eqref{K1} for the Ginibre product  process.

(b) As $b\rightarrow 0$, also for $r=m$ or $s=m$ the correlation kernel $K_{n,m}(r,x;s,y;b)$ turns into the correlation kernel $K_{n,m}^{\Ginibre}(r,x;s,y)$ of the Ginibre product process (\ref{K1}) (with $m-1\to m$) from
\cite[Prop.2.3]{StrahovD}.
\end{Remarks}

As a particular case (corresponding to $r=s=m$) we obtain from Theorem \ref{TheoremDoubleIntegralRepresentationExactKernel}
a double contour integral representation for the correlation kernel of the biorthogonal ensemble defined in Proposition \ref{PropositionJointDensityTotalProductMatrix}, where only the squared singular values of the total product matrix $Y_m$ are retained.
\begin{prop}
Consider the biorthogonal ensemble of Proposition \ref{PropositionJointDensityTotalProductMatrix}.
This biorthogonal ensemble defines a determinantal point process on $\R_{>0}$. The correlation kernel of
this determinantal point process
can be written as
\begin{equation}\label{KERNEL}
\begin{split}
K_{n,m}(x;y;b)=&\frac{1}{(2\pi i)^2}\int\limits_{-\frac{1}{2}-i\infty}^{-\frac{1}{2}+i\infty}du
\oint\limits_{\Sigma_n}dt\frac{\Gamma(t-n+1)}{\Gamma(u-n+1)}
\frac{\prod_{j=0}^m\Gamma(u+\nu_j+1)}{\prod_{j=0}^m\Gamma(t+\nu_j+1)}\\
&\times\frac{x^ty^{-u-1}}{u-t}\frac{2\Gamma(t+1)\left(by^{\frac12}\right)^{u+1}}{\Gamma(u+1)\left(bx^{\frac12}\right)^t}
I_t\left(2bx^{\frac12}\right)K_{u+1}\left(2by^{\frac12}\right).
\end{split}
\end{equation}
\end{prop}

As $b\rightarrow 0$,
due to eqs. \eqref{ISmallb} and \eqref{KSmallb}
the last factors in the second line of
\eqref{KERNEL} can be replaced by 1, and we obtain
\begin{equation}
\begin{split}
&\underset{b\rightarrow 0}{\lim}K_{n,m}(x,y;b)=\frac{1}{(2\pi i)^2}\int\limits_{-\frac{1}{2}-i\infty}^{-\frac{1}{2}+i\infty}du\oint\limits_{\Sigma_n}dt
\frac{\Gamma(t-n+1)}{\Gamma(u-n+1)}
\frac{\prod_{j=0}^m\Gamma(u+\nu_j+1)}{\prod_{j=0}^m\Gamma(t+\nu_j+1)}\frac{x^ty^{-u-1}}{u-t}.
\end{split}
\end{equation}
The expression in the right-hand side of this equation 
represents the correlation kernel
for the ensemble of the squared singular values of $m$ rectangular matrices with independent complex Gaussian entries,
see Kuijlaars and Zhang \cite[Prop.5.1]{KuijlaarsZhang}.

\subsection{Hierarchy of correlation kernels at finite matrix size}
In order to summarise the results presented above let us briefly describe the hierarchy of the correlation kernels under considerations.
On top of the hierarchy we have the correlation kernel  $K_{n,m}^{V}\left(r,x;s,y\right)$ \eqref{KnmV} of Theorem \ref{TheoremMostGerneralCorrelationKernel}.
This kernel depends on $m$ potential functions $V_1(t)$, $\ldots$, $V_m(t)$, that couple the different levels, and lives on the space
$\left\{1,\ldots,m\right\}\times\R_{>0}$.

The correlation kernel $K_{n,m}\left(r,x;s,y;b\right)$ \eqref{CorrelationKernelFormulaContour} of Theorem \ref{TheoremDoubleIntegralRepresentationExactKernel}
is the specialisation of $K_{n,m}^{V}\left(r,x;s,y\right)$ to the case where $V_1(t)=\ldots=V_{m-2}(t)=0$, 
$V_{m-1}(t)=b^2t$, and $V_m(t)=0$.
It depends on one coupling constant $b>0$, and defines the Ginibre product process with coupling.
As we pointed out already, due to our choice of potentials this point process agrees with the  Ginibre product process without coupling for the first $m-1$ levels, with $r,s<m$.
On the other hand, taking $r=s=m$, we obtain the kernel
$K_{n,m}(x,y;b)$ 
that depends on the parameter $b$. It describes the biorthogonal ensemble for the singular values of the total product matrix with coupling.
Thus, for $r=m$ or $s=m$
the kernel $K_{n,m}\left(r,x;s,y;b\right)$ coupled to the $m$th level can be understood as a deformation of the kernel $K_{n,m}^{\Ginibre}(r,x;s,y)$ of the Ginibre product process.
For $r=s=m$, the latter kernel $K_{n,m}^{\Ginibre}(m,x;m,y)$ specialises to the finite-$n$ Meijer $G$-kernel, which is the correlation kernel for the ensemble of the squared singular values of $m$ rectangular matrices with independent complex Gaussian entries. Note that  the finite-$n$ Meijer $G$-kernel
can be obtained from $K_{n,m}(x,y;b)$ by taking the coupling parameter $b$ to  zero. Finally, if $r=s=1$, the kernel
$K_{n,m}^{\Ginibre}(1,x;1,y)$ specialises to that of the classical Laguerre ensemble.
\section{Hard edge scaling limits of the multi-level determinantal processes}\label{SectionHardEdgeScalingLimits}
The point  processes considered in this paper
are uniquely determined by their correlation functions. Therefore, we will say
that the point processes $\Pc$ converges to the point process  $\Pc'$, if all correlation functions of $\Pc$
converge to the corresponding correlation functions of $\Pc'$. Since we are dealing
with determinantal point processes only, the convergence of the correlation kernel of $\Pc$ to the correlation kernel of $\Pc'$
will be considered as equivalent to the convergence of $\Pc$ to $\Pc'$.

In order to discuss scaling limits of   the multi-level determinantal processes
it is convenient to introduce the following notation, following \cite{StrahovD}. Denote by $K_{\infty, m}^{\Ginibre}(r,x;s,y)$ the kernel given by the formula
\begin{equation}\label{KernelInfiniteGinibreProductProcess}
\begin{split}
K_{\infty, m}^{\Ginibre}(r,x;s,y)=&-\frac{1}{x}G^{s-r,0}_{0,s-r}\left(\begin{array}{ccc}
                - \\
               \nu_{r+1},\ldots,\nu_s
             \end{array}\biggl|\frac{y}{x}\right)\mathbf{1}_{s>r}\\
&+\frac{1}{(2\pi i)^2}\int\limits_{-\frac{1}{2}-i\infty}^{-\frac{1}{2}+i\infty}
du\oint\limits_{\Sigma_{\infty}}dt\frac{\prod_{j=0}^s\Gamma(u+\nu_j+1)}{\prod_{j=0}^r\Gamma(t+\nu_j+1)}
\frac{\sin\pi u}{\sin\pi t}
\frac{x^ty^{-u-1}}{u-t}.
\end{split}
\end{equation}
Here $x>0$, $y>0$, the parameters  $r, s$ take values in $\left\{1,\ldots,m\right\}$, and  $\Sigma_{\infty}$  is a contour starting from $+\infty$ in the upper half plane and returning to $+\infty$ in the lower half plane, leaving $-\frac{1}{2}$ on the left, and encircling $\{0,1,2,\ldots \}$.
Here, we write the correlation kernel $K_{\infty,m}^{\Ginibre}(r,x;s,y)$ with an index $m$, only to emphasise that the variables $r$ and $s$ take values in $\left\{1,\ldots,m\right\}$.
Note that for $r=s$ the kernel $K_{\infty,m}^{\Ginibre}(r,x;r,y)=K_{\infty,r}^{\Ginibre}(r,x;r,y)$ is the limiting \textit{Meijer G-kernel}  obtained by Kuijlaars and Zhang \cite{KuijlaarsZhang}. It describes
the hard edge scaling limit for the product of $r$ independent complex Gaussian matrices and reduces to the standard Bessel kernel for $r=1$.
In \cite{StrahovD} Strahov showed that $K_{\infty,m}^{\Ginibre}(r,x;s,y)$ can be understood as a multi-level extension of the infinite Meijer $G$-kernel.

In the context of products of random matrices, the kernel $K_{\infty, m}^{\Ginibre}(r,x;s,y)$ describes  the hard edge  scaling limit of the Ginibre product process. Recall
that the Ginibre product process is a determinantal point process on
$\left\{1,\ldots,m\right\}\times\R_{>0}$ whose density is proportional to expression (\ref{GinibreDensityFormula}),
and whose correlation kernel is given by formula (\ref{K1}).
The following proposition was shown in \cite{StrahovD}:
\begin{prop}\label{PropositionScalingLimitGInibreProcess} Assume that the point configurations $\left(y_1^m,\ldots,y_n^m;\ldots;y_1^1,\ldots,y_n^1\right)$
form the Ginibre product process on $\left\{1,\ldots,m\right\}\times\R_{>0}$, given by (\ref{GinibreDensityFormula}). As $n\rightarrow\infty$,
the scaled Ginibre process  formed by the point configurations $\left(ny_1^m,\ldots,ny_n^m;\ldots;ny_1^1,\ldots,ny_n^1\right)$
converges to the determinantal point process
$\Pc^{\Ginibre}_{\infty,m}$ on $\left\{1,\ldots,m\right\}\times\R_{>0}$, whose correlation kernel,
$K_{\infty, m}^{\Ginibre}(r,x;s,y)$, is defined by equation (\ref{KernelInfiniteGinibreProductProcess}).

Alternatively, we have
\begin{equation}
\label{GinlimES}
\underset{n\rightarrow\infty}{\lim}\frac{1}{n}K_{n,m}^{\Ginibre}\left(r,\frac{x}{n};s,\frac{y}{n}\right)
=K_{\infty, m}^{\Ginibre}(r,x;s,y),
\end{equation}
where $K_{n, m}^{\Ginibre}(r,x;s,y)$ is defined by equation (\ref{K1})
and $K_{\infty, m}^{\Ginibre}(r,x;s,y)$
by equation (\ref{KernelInfiniteGinibreProductProcess}). The variables $x$ and $y$ take values in a compact subset of $\R_{>0}$, and $1\leq r,s\leq m$.
\end{prop}

Now, let us consider the convergence of the Ginibre product process with coupling. Recall that by
this
we mean the multi-level determinantal process on $\left\{1,\ldots,m\right\}\times\R_{>0}$
whose density is given by expression (\ref{CouplingDensityFormula1})
and whose correlation kernel $K_{n,m}(r,x;s,y;b)$ is given by Theorem \ref{TheoremDoubleIntegralRepresentationExactKernel}.
As we already pointed out previously, for $r,s\leq m-1$ this kernel agrees with that of the Ginibre product process \textit{without} coupling, \eqref{K1}, for which the hard edge limit was already worked out in \cite{StrahovD}, see \eqref{GinlimES} in Proposition \ref{PropositionScalingLimitGInibreProcess} above (with $m-1\to m$). Therefore, in the following
theorem
we will only find nontrivial results for
the limit of the kernel $K_{n,m}(r,x;s,y;b)$ with $r=m$ or $s=m$, that includes  the $m$-th level.
In particular
we will now drop the assumption that the coupling parameter $b$ is  constant, and consider $b$ as
a function of the number of particles on the same level, $b=b(n)$.
Recall that once $b$ approaches zero, the Ginibre product process with coupling turns
into
that without,
as discussed in Section \ref{SectionExactResults}.
Depending on how  $b(n)$ behaves as a function of $n$, we find three different limits for the limiting kernel that includes level $m$.

\begin{thm}\label{TheoremHardEdgeGinibreCouplingProcess}
Let the point configurations $\left(y_1^m,\ldots,y_n^m;\ldots;y_1^1,\ldots,y_n^1\right)$ with density  (\ref{CouplingDensityFormula1})
form the Ginibre product process with coupling on $\left\{1,\ldots,m\right\}\times\R_{>0}$, with correlation kernel $K_{n,m}(r,x;s,y;b)$
of Theorem \ref{TheoremDoubleIntegralRepresentationExactKernel}, equation (\ref{CorrelationKernelFormulaContour}).\\
\textbf{(A)} \textbf{Weak coupling regime}. Assume that ${b(n)}/{\sqrt{n}}\rightarrow 0$, as $n\rightarrow\infty$. The scaled Ginibre product process with coupling
of configurations $\left(ny_1^m,\ldots,ny_n^m;\ldots;ny_1^1,\ldots,ny_n^1\right)$ converges as $n\rightarrow\infty$ to  the determinantal point process
$\Pc^{\Ginibre}_{\infty,m}$.
Equivalently,
we have the following relation to the limiting kernel of the Ginibre product process (\ref{KernelInfiniteGinibreProductProcess}), for all levels
$1\leq r,s\leq m$:
\begin{equation}\label{WeakCouplingRegimeLimit}
\underset{n\rightarrow\infty}{\lim}\frac{1}{n}K_{n,m}\left(r,\frac{x}{n};s,\frac{y}{n};b(n)\right)
=K_{\infty, m}^{\Ginibre}(r,x;s,y).
\end{equation}
Here,
the variables $x$ and $y$ take values in a compact subset of $\R_{>0}$.\\
\textbf{(B)} \textbf{Interpolating
regime}. Assume that ${b(n)}/{\sqrt{n}}\rightarrow \alpha$, as $n\rightarrow\infty$, where $\alpha>0$.
As $n\rightarrow\infty$,  the scaled Ginibre product process with coupling
formed by the point configurations $\left(ny_1^m,\ldots,ny_n^m;\ldots;ny_1^1,\ldots,ny_n^1\right)$
converges to the determinantal point process $\Pc^{\interpol}_{\infty,m}(\alpha)$ on $\left\{1,\ldots,m\right\}\times\R_{>0}$. Its correlation kernel
is given by
\begin{equation}\label{KernelZKExtended}
\begin{split}
&K_{\infty, m}^{\interpol}(r,x;s,y;\alpha)=-\phi_{r,s}(x,y;\alpha)\\
             &+\frac{1}{(2\pi i)^2}\int\limits_{-\frac{1}{2}-i\infty}^{-\frac{1}{2}+i\infty}
du\oint\limits_{\Sigma_{\infty}}dt\frac{\prod_{j=0}^s\Gamma(u+\nu_j+1)}{\prod_{j=0}^r\Gamma(t+\nu_j+1)}
\frac{\sin\pi u}{\sin\pi t}
\frac{x^tp_r(t,x;\alpha)q_s(u+1,y;\alpha)}{y^{u+1}(u-t)},
\end{split}
\end{equation}
for $r=m$ or $s=m$, with
the functions $\phi_{r,s}(x,y;\alpha)$, $p_r(t,x;\alpha)$, and $q_s(u+1,y;\alpha)$ defined  by equations (\ref{PHI(x,y,b)}), (\ref{Functionp}), and (\ref{Functionq}), respectively.
For $1\leq r,s\leq m-1$ its correlation kernel is given by $K_{\infty, m-1}^{\Ginibre}(r,x;s,y)$.
Equivalently, we have for $r=m$ or $s=m$:
\begin{equation}
\underset{n\rightarrow\infty}{\lim}\frac{1}{n}K_{n,m}\left(r,\frac{x}{n};s,\frac{y}{n};b(n)\right)
=K_{\infty, m}^{\interpol}(r,x;s,y;\alpha),
\end{equation}
and \eqref{WeakCouplingRegimeLimit} (with $m\to m-1$ on the right-hand side)
for $1\leq r,s\leq m-1$.
Here,
the variables $x$ and $y$ take values in a compact subset of $\R_{>0}$.\\
\textbf{(C)} \textbf{Strong coupling regime}. Assume that ${b(n)}/{\sqrt{n}}\rightarrow\infty$, as $n\rightarrow\infty$.
The scaled  process Ginibre product process with coupling
formed by the point configurations
$$
\left(\frac{n}{b(n)}\left(y_1^m\right)^{\frac{1}{2}},\ldots,\frac{n}{b(n)}\left(y_n^m\right)^{\frac{1}{2}};
ny_1^{m-1},\ldots,ny_n^{m-1};
\ldots;
ny_1^1,\ldots,ny_n^1\right)
$$
converges as $n\rightarrow\infty$  to the determinantal point process $\Pc^{\Ginibre}_{\infty,m-1}$ of $m-1$ levels, with the following
identification of levels $m-1$ and $m$ on $\left\{1,\ldots,m\right\}\times\R_{>0}$.
Namely, its  correlation kernel
is given by the following limiting relations:\\
$\bullet$ For $1\leq r,s\leq m-1$ we have equation \eqref{WeakCouplingRegimeLimit} (with  $m\to m-1$ on the right-hand side).
$\bullet$ For $1\leq r\leq m-1$, $s=m$ we have
\begin{equation}\label{SKR2}
\begin{split}
&\underset{n\rightarrow\infty}{\lim}\frac{2^{\frac{1}{2}}b(n)y^{\frac{1}{2}}}{n^{\frac{3}{2}}}
K_{n,m}\left(r,\frac{x}{n};m,\frac{b^2(n)y^2}{n^2};b(n)\right)\frac{e^{2\frac{b^2(n)}{n}y}}{2^{\frac{1}{2}}\pi^{\frac{1}{2}}}\\
&=K^{\Ginibre}_{\infty,m-1}(r,x;m-1,y)+\mathbf{1}_{r,m-1}\delta(x-y).
\end{split}
\end{equation}
$\bullet$
For $r=m$, $1\leq s\leq m-1$  we have
\begin{equation}\label{SKR3}
\underset{n\rightarrow\infty}{\lim}\frac{2^{\frac{1}{2}}b(n)x^{\frac{1}{2}}}{n^{\frac{3}{2}}}
K_{n,m}\left(m,\frac{b^2(n)x^2}{n^2};s,\frac{y}{n};b(n)\right)\frac{2^{\frac{1}{2}}\pi^{\frac{1}{2}}}{e^{2\frac{b^2(n)}{n}x}}
=K^{\Ginibre}_{\infty,m-1}(m-1,x;s,y).
\end{equation}
$\bullet$
For $r=m$ and $s=m$ we have
\begin{equation}\label{SKR4}
\begin{split}
&\underset{n\rightarrow\infty}{\lim}\frac{2b^2(n)x^{\frac{1}{2}}y^{\frac{1}{2}}}{n^{2}}
K_{n,m}\left(m,\frac{b^2(n)x^2}{n^2};m,\frac{b^2(n)y^2}{n^2};b(n)\right)
\frac{e^{2\frac{b^2(n)}{n}y}}{e^{2\frac{b^2(n)}{n}x}}\\
&=K^{\Ginibre}_{\infty,m-1}(m-1,x;m-1,y).
\end{split}
\end{equation}
In all these formulas the variables $x$ and $y$ take values in a compact subset of $\R_{>0}$.
\end{thm}
\begin{Remarks}
(a) The interpolating point process $\Pc^{\interpol}_{\infty,m}(\alpha)$, that contains
the correlation kernel $K_{\infty, m}^{\interpol}(r,x;s,y;\alpha)$ for $r$ or $s=m$, can be understood as a one-parameter 
{deformation}
of the determinantal point process $\Pc^{\Ginibre}_{\infty,m}$
defined by  the correlation kernel $K_{\infty, m}^{\Ginibre}(r,x;s,y)$. This will be made more precise in the next Section \ref{SectionInterpolatingProcess}.\\
(b) The limiting point process 
obtained
in the strong coupling
regime is a determinantal process living on $m-1$ levels, with an infinite number of points on each level, and where the levels $m$ and $m-1$ have been identified.
The correlations between particles on the first $m-1$ levels trivially agrees
with that of $\Pc^{\Ginibre}_{\infty,m-1}$, the hard edge scaling limit of the Ginibre product process with $m-1$ levels.
The correlations between particles only on the $m$-th level agree with those of the Ginibre product process on the $(m-1)$-th level.
Furthermore, the correlations between particles on the $m$-th level and all other particles on the first $m-1$ levels are the same as if all the particles on the $m$-th level would be on the $(m-1)$-th level - unless two points on level $m$ and level $m-1$ coincide, hence the contact interaction in terms of the Dirac delta.
In other words, in the asymptotic strong coupling  regime the point process on the $m$-th level
\textit{collapses} 
to that
on the $(m-1)$-th level, and the $m$-matrix  model behaves statistically like an $(m-1)$-matrix model.\\
(c) In particular, for $m=2$ the limiting process at strong coupling collapses to
the classical Bessel-kernel process of a single level, with the collapse of the first to the second level described in the previous remark.
\end{Remarks}
For the biorthogonal ensemble formed by the singular values of the total product matrix $Y_m=G_m\cdots G_1$ only,
Theorem \ref{TheoremHardEdgeGinibreCouplingProcess} leads to the following result.
\begin{cor}\label{CorollaryTheoremHardEdgeGinibreCouplingProcess} Let $Y_m=G_m\cdots G_1$ be the total product matrix of the multi-matrix model defined by equation (\ref{MainProbabilityMeasure1}),
and let $\left(y_1,\ldots,y_n\right)$ be the squared singular values of $Y_m$.\\
\textbf{(A)} Assume that ${b(n)}/{\sqrt{n}}\rightarrow 0$, as $n\rightarrow\infty$.
Then, the point process formed by the point configurations
$\left(ny_1,\ldots,ny_n\right)$ converges to the process for the product of $m$ independent Gaussian matrices,
i.e. to the determinantal process on $\R_{>0}$ with the correlation kernel given by the Meijer $G$-kernel $K_{\infty, m}^{\Ginibre}(m,x;m,y)$,
equation (\ref{KernelInfiniteGinibreProductProcess}).\\
\textbf{(B)}  Assume that ${b(n)}/{\sqrt{n}}\rightarrow\alpha$, as $n\rightarrow\infty$, where $\alpha>0$.
Then, the point process formed by the point configurations
$\left(ny_1,\ldots,ny_n\right)$ converges to the determinantal process on $\R_{>0}$ defined by the correlation  kernel $K_{\infty, m}^{\interpol}(m,x;m,y;\alpha)$,
given by equation (\ref{KernelZKExtended}). The limiting determinantal process can be understood as
a one-parameter
deformation of the Meijer $G$-kernel process for the product of $m$ independent Gaussian matrices.\\
\textbf{(C)}  Assume that ${b(n)}/{\sqrt{n}}\rightarrow\infty$, as $n\rightarrow\infty$.
Then, the scaled  process
formed by the point configurations
$
\left(\frac{n}{b(n)}\left(y_1^m\right)^{\frac{1}{2}},\ldots,\frac{n}{b(n)}\left(y_n^m\right)^{\frac{1}{2}}\right)
$
converges to the determinantal process on $\R_{>0}$ with the correlation kernel given by the Meijer $G$-kernel $K_{\infty, m-1}^{\Ginibre}(m-1,x;m-1,y)$,
the limiting
process for the product of $m-1$ independent Gaussian matrices.

\end{cor}
\section{The interpolating multi-level determinantal process}\label{SectionInterpolatingProcess}
Consider the determinantal process  $\Pc^{\interpol}_{\infty,m}(\alpha)$ on $\left\{1,\ldots,m\right\}\times\R_{>0}$ in the interpolating regime.
We will show that it 
interpolates between the hard edge scaling limit $\Pc^{\Ginibre}_{\infty,m}$
of the Ginibre product process of $m$ independent levels,
and that
in the strong coupling regime among $m-1$ levels.
For this reason we will refer to $\Pc^{\interpol}_{\infty,m}(\alpha)$ as to the interpolating multi-level
determinantal process.
\begin{thm}\label{TheoremInterpolation}
Assume that the point configurations
$$
\left(x_1^m,x_2^m,\ldots;x_1^{m-1},x_2^{m-1},\ldots;\ldots;x_1^1,x_2^1,\ldots\right)
$$
form the determinantal process  $\Pc^{\interpol}_{\infty,m}(\alpha)$ on $\left\{1,\ldots,m\right\}\times\R_{>0}$
in the interpolating regime.\\
\textbf{(A)} As $\alpha\rightarrow 0$, this determinantal process  $\Pc^{\interpol}_{\infty,m}(\alpha)$ converges
to the determinantal process  $\Pc^{\Ginibre}_{\infty,m}$ corresponding to the hard edge scaling limit
of the Ginibre product process of $m$ levels. In particular, for $r=m$ or $s=m$ we find
$$
\underset{\alpha\rightarrow 0}{\lim}K_{\infty, m}^{\interpol}(r,x;s,y;\alpha)=K_{\infty, m}^{\Ginibre}(r,x;s,y).
$$
where the variables $x$ and $y$ take values in a compact subset of $\R_{>0}$.\\
\textbf{(B)} Consider the scaled interpolating process formed by the point configurations
$$
\left(\frac{1}{\alpha}\left(x_1^m\right)^{\frac{1}{2}},
\frac{1}{\alpha}\left(x_2^m\right)^{\frac{1}{2}},\ldots;x_1^{m-1},x_2^{m-1},\ldots;\ldots;x_1^1,x_2^1,\ldots\right).
$$
As $\alpha\rightarrow\infty$, this determinantal process converges to
the hard edge scaling limit
in the strong coupling regime.

This is equivalent to the following statement:\\
$\bullet$ For $1\leq r,s\leq m-1$ we have the kernel \eqref{KernelInfiniteGinibreProductProcess}
(with  $m\to m-1$ on the right-hand side).\\
$\bullet$ For $1\leq r\leq m-1$, $s=m$ we have
\begin{equation}\label{I2}
\underset{\alpha\rightarrow\infty}{\lim}\left(2^{\frac{1}{2}}\alpha y^{\frac{1}{2}}
K_{\infty, m}^{\interpol}(r,x;m,\alpha^2y^2;\alpha)
\frac{e^{2\alpha^2y^2}}{2^{\frac{1}{2}}\pi^{\frac{1}{2}}}
\right)=
K_{\infty, m}^{\Ginibre}(r,x;m-1,y)+\mathbf{1}_{r,m-1}\delta(x-y).
\end{equation}
$\bullet$ For $r=m$, $1\leq s\leq m-1$ we have
\begin{equation}\label{I3}
\underset{\alpha\rightarrow\infty}{\lim}\left(2^{\frac{1}{2}}\alpha x^{\frac{1}{2}}
K_{\infty, m}^{\interpol}(m,\alpha^2x^2;s,y;\alpha)
\frac{2^{\frac{1}{2}}\pi^{\frac{1}{2}}}{e^{2\alpha^2x^2}}
\right)=
K_{\infty, m}^{\Ginibre}(m-1,x;s,y).
\end{equation}
$\bullet$ Finally,  for $r=m$, $s=m$ we have
\begin{equation}\label{I4}
\underset{\alpha\rightarrow\infty}{\lim}\left(2\alpha^2x^{\frac{1}{2}}y^{\frac{1}{2}}
K_{\infty, m}^{\interpol}(m,\alpha^2x^2;m,\alpha^2y^2;\alpha)\frac{e^{2\alpha^2y^2}}{e^{2\alpha^2x^2}}\right)=
K_{\infty, m}^{\Ginibre}(m-1,x;m-1,y).
\end{equation}
In all these formulas the variables $x$ and $y$ take values in a compact subset of $\R_{>0}$.
\end{thm}
In particular, Theorem \ref{TheoremInterpolation} implies that the determinantal  process on $\R_{>0}$ defined by the correlation kernel
$K_{\infty, m}^{\interpol}(m,x;m,y;\alpha)$ is an interpolating (one-level) determinantal point process. It interpolates between
the Meijer $G$-kernel process for $m$ independent matrices, and the Meijer $G$-kernel process for $m-1$ independent matrices (with contact interaction from the identification of level $m$  with level $m-1$).
\section{An integration formula for coupled matrices}\label{SectionIntegrationFormula}
Below we derive an integration formula related to the investigation of singular variables of coupled matrices, see Lemma \ref{LemmaCoupling}.
The  obtained formula will be applied to multi-matrix models, in order to derive the corresponding joint densities.
\begin{lem}\label{LemmaCoupling}
Assume that $\nu\geq 0$ is an integer, and $b>0$. Let
$G$ be a complex matrix of size $(n+\nu)\times l$,
where $l\geq n\geq 1$,
and $X$ be a complex matrix of size $l\times n$,
with non-zero squared singular values
$x_1,\ldots, x_n$. 
Denote by $y(GX)=\{y_1,\ldots, y_n\}$
the set of the squared singular values  of matrix $Y=GX$.
The complex flat Lebesgue measure for matrix $G$ is denoted by $dG$, i.e.
$dG=\prod_{j=1}^l\prod_{k=1}^{n+\nu}dG_{j,k}^RdG_{j,k}^I$,  $G_{j,k}=G_{j,k}^R+iG_{j,k}^I$  denotes the sum
of the real and imaginary parts of the matrix entries $G_{j,k}$,
Let $f(.)$ be a continuous function on
$\R_{>0}^n$ with compact support. In addition, we assume that $f(y)=f(y_1,\ldots,y_n)$ is symmetric with respect to
permutations of $y_1$, $\ldots$, $y_n$.
Finally, $V$ (the potential) is some scalar positive function, such that all the following three matrix integrals exist.
\\
\textbf{(A)} We have for $l=n$
\begin{equation}\label{CouplingFormula}
\begin{split}
&\int f(y(GX))e^{-\Tr\left[G^*G\right]+b\Tr\left[GX+\left(GX\right)^*\right]-\Tr V\left(\left(GX\right)^*(GX)\right)}dG\\
&=\frac{c}{\Delta_n(\{x_j\})}\\
&\times\int\limits_0^{\infty}\ldots\int\limits_0^{\infty}f(y_1,\ldots,y_n)
\det\left[\frac{e^{-\frac{y_j}{x_k}-V(y_j)}}{x_k}\right]_{j,k=1}^n\det\left[y_j^{\frac{k-1}{2}}I_{k-1}
\left(2by_j^{\frac{1}{2}}\right)\right]_{j,k=1}^ndy_1\ldots dy_n,
\end{split}
\end{equation}
where we recall that the Vandermonde determinant was defined in \eqref{Vandermonde}.
The constant $c$ does not depend on the set $\{x_1,\ldots,x_n\}$. \\
\textbf{(B)} Without the coupling term ($b=0$) and for general $l\geq n$ we have
\begin{equation}\label{CouplingFormula1}
\begin{split}
&\int f(y(GX))e^{-\Tr\left[G^*G\right]-\Tr V\left(\left(GX\right)^*(GX)\right)}dG\\
&=\frac{c'}{\Delta_n(\{x_i\})}\int\limits_0^{\infty}\ldots\int\limits_0^{\infty}f(y_1,\ldots,y_n)
\det\left[\frac{y_j^{\nu}e^{-\frac{y_j}{x_k}-V(y_j)}}{x_k^{\nu+1}}\right]_{j,k=1}^n\Delta_n\left(\{y_i\}\right)dy_1\ldots dy_n.
\end{split}
\end{equation}
Here the constant $c'$ does not depend on the set $\{x_1,\ldots,x_n\}$.\\
\textbf{(C)} Without matrix $X$ and for $l=n$ we have:
\begin{equation}\label{CouplingFormula2}
\begin{split}
&\int f(y(G))e^{-\Tr\left[G^*G\right]-\Tr V\left(G^*G\right)}dG\\
&=c''\int\limits_0^{\infty}\ldots\int\limits_0^{\infty}f(y_1,\ldots,y_n)
\Delta_n^2\left(\{y_i\}\right)\prod\limits_{j=1}^ny_j^{\nu}e^{-y_j-V(y_j)}dy_1\ldots dy_n,
\end{split}
\end{equation}
where here $y(G)=\{y_1,\ldots,y_n\}$ is the set of squared singular values of $G$ instead.
\end{lem}
\begin{proof}
Consider the following measure
$$
P(G,X)dG=e^{-\Tr\left[G^*G\right]+b\Tr\left[GX+\left(GX\right)^*\right]-\Tr V\left(\left(GX\right)^*(GX)\right)}dG.
$$
Following the analysis of Fischman et al. \cite{Jonit} we set
$$
X=U\left(\begin{array}{c}
           X_0 \\
           0
         \end{array}
\right),
$$
where $U$ is an $(n+\nu)\times (n+\nu)$ unitary matrix, $0$ is a $\nu\times n$ matrix consisting of zeros only, and $X_0$ is an $n\times n$ complex matrix.
We have
\begin{equation}
\begin{split}
&P(G,X)dG\\
&=e^{-\Tr\left[G^*G\right]+b\Tr\left[GU\left(\begin{array}{c}
           X_0 \\
           0
         \end{array}
\right)
+\left(\begin{array}{cc}
                                X_0^* & 0
                              \end{array}
\right)U^*G^*\right]-\Tr V\left(
                              \left(\begin{array}{cc}
                                      X_0^* & 0
                                    \end{array}
                              \right)U^*G^*
GU\left(\begin{array}{c}
                                X_0\\
                                0
                              \end{array}\right)
                              \right)}dG.
\end{split}
\nonumber
\end{equation}
Set $\widehat{G}=GU$, and note that $\widehat{G}$ is a matrix of size $n\times (n+\nu)$.
Use the invariance of the Lebesgue measure under unitary transformations to write
$$
P(G,X)dG=e^{-\Tr\left[\widehat{G}^*\widehat{G}\right]+b\Tr\left[\widehat{G}\left(\begin{array}{c}
           X_0 \\
           0
         \end{array}
\right)
+\left(\begin{array}{cc}
                                X_0^* & 0
                              \end{array}
\right)\widehat{G}^*\right]
-\Tr V\left(
                              \left(\begin{array}{cc}
                                      X_0^* & 0
                                    \end{array}
                              \right)\widehat{G}^*
\widehat{G}\left(\begin{array}{c}
                                X_0\\
                                0
                              \end{array}\right)
                              \right)}d\widehat{G},
$$
where $d\widehat{G}=\prod_{j=1}^n\prod_{k=1}^{n+\nu}d\widehat{G}_{j,k}^Rd\widehat{G}_{j,k}^I$, and where
$\widehat{G}_{j,k}=\widehat{G}_{j,k}^R+i\widehat{G}_{j,k}^I$  denotes the sum
of the real and imaginary parts of the matrix entries $\widehat{G}_{j,k}$. Write
$\widehat{G}$ as $\widehat{G}=\left(G_0,G_1\right)$, where $G_0$ is a matrix of size $n\times n$, and where
$G_1$ is a matrix of size $n\times\nu$. We obtain the following decomposition of the measure $P(G,X)dG$
\begin{equation}
P(G,X)dG=\left(e^{-\Tr\left[G_0^*G_0\right]+b\Tr\left[G_0X_0
+G_0^*X_0^*\right]
-\Tr V\left(\left(G_0X_0\right)^*\left(G_0X_0\right)\right)}dG_0\right)
e^{-\Tr\left[G_1^*G_1\right]}dG_1.
\end{equation}
We have put brackets here to emphasise that the integrals over $G_0$ and $G_1$ decouple, with the latter giving only an additional multiplicative constant.
The important observation is that the matrices $X$ and $X_0$ have the same singular values, and that the matrices $GX$ and $G_0X_0$ have the same singular values.
Set $Y=G_0X_0$. We have  $dG_0=\det^{-n}\left[X_0^*X_0\right]dY$, which gives
\begin{equation}
\begin{split}
&P(G,X)dG\\
&=\left(e^{-\Tr\left[\left(YX_0^{-1}\right)^*YX_0^{-1}\right]+b\Tr\left[Y
+Y^*\right]
-\Tr V\left(Y^*Y\right)}
{\det}^{-n}\left[X_0^*X_0\right]dY\right)
e^{-\Tr\left[G_1^*G_1\right]}dG_1.
\end{split}
\end{equation}
Here, we have used that $X_0$ is invertible, which is ensured by its non-zero squared singular values.

The the singular value decomposition of the matrix $Y$ can be written as $Y=\widetilde{U}\Sigma P$, where both $\widetilde{U}$, $P$ are unitary matrices of the same size $n\times n$,
and where $\Sigma$ is an $n\times n$ diagonal matrix with a real matrix entries,
$$
\Sigma=\left(\begin{array}{cccc}
               \sqrt{y_1} & 0 & \ldots & 0 \\
               0 & \sqrt{y_2} & \ldots & 0 \\
               \vdots &  &  &\\
               0 & 0 & 0 & \sqrt{y_n}
             \end{array}
\right),
$$
and $y_1,\ldots,y_n$ are the
squared singular values of $Y$.
The Jacobian 
corresponding to this decomposition is
$$
dY=c_1\Delta_n(\{y_j\})^2d\widetilde{U}dPdy_1\ldots dy_n,
$$
where $c_1$ is some constant. We thus obtain
\begin{equation}
\begin{split}
P(G,X)dG=&c_1e^{-\Tr\left[P^*\Sigma^2 P\left(X_0^*X_0\right)^{-1}\right]+b\Tr\left[\widetilde{U}\Sigma P
+P^*\Sigma \widetilde{U}^*\right]
-\Tr V\left(\Sigma^2\right)}
{\det}^{-n}\left[X_0^*X_0\right]\Delta_n(\{y_j\})^2d\widetilde{U}dV\\
&\times e^{-\Tr\left[G_1^*G_1\right]}dG_1.
\end{split}
\nonumber
\end{equation}
The next step is to use the invariance of the Haar measure under left shifts by the group elements, $\widetilde{U}\rightarrow \hat{U}=P\widetilde{U}$, and to rewrite
the expression above as
\begin{equation}
\begin{split}
P(G,X)dG=&c_1e^{-\Tr\left[P^*\Sigma^2 P\left(X_0^*X_0\right)^{-1}\right]+b\Tr\left[\hat{U}\Sigma
+\Sigma \hat{U}^*\right]-\Tr V\left(\Sigma^2\right)}
{\det}^{-n}\left[X_0^*X_0\right]\Delta_n(\{y_j\})^2d\hat{U}dP\\
&\times e^{-\Tr\left[G_1^*G_1\right]}dG_1.
\end{split}
\nonumber
\end{equation}

The integration over $P$ can be performed using the Harish-Chandra--Itzykson-Zuber integral formula \cite{HC,IZ}
$$
\int_{U(n)}e^{-\Tr\left[P^*\Sigma^2P\left(X_0^*X_0\right)^{-1}\right]}dP=\const\frac{\det\left[e^{-\frac{y_j}{x_k}}\right]_{j,k=1}^n}{\triangle(\{y_j\})\Delta_n(\{x_j^{-1}\})},
$$
where the constant depends only on $n$. The integration over ${\hat{U}}$ can be done exploiting the following integration formula \cite{Brower} (sometimes called Leutwyler-Smilga formula \cite{LeutwylerSmilga})
$$
\int_{U(n)}e^{b\Tr\left[\Sigma\left({\hat{U}}+{\hat{U}}^*\right)\right]}d{\hat{U}}
=\const\frac{\det\left[y_j^{\frac{k-1}{2}}I_{k-1}\left(2by_j^{\frac{1}{2}}\right)\right]_{j,k=1}^n}{\Delta_n(\{y_j\})},
$$
where $I_{\kappa}(x)$ denotes the modified Bessel function of the first kind.
After integration, and after some simplifications,
$\Delta_n(\{x_j^{-1}\})=(-1)^{\frac{n(n-1)}{2}}\prod_{k=1}^nx_k^{-n+1}
\Delta_n(\{x_j\})$,
we obtain  formula (\ref{CouplingFormula}) in the statement of the Lemma. Formula (\ref{CouplingFormula1}) can be obtained in the same way, and formula (\ref{CouplingFormula2})
is well known.
\end{proof}
\begin{prop}\label{TheoremDeformedDensity} Let $\nu$, $X$, $G$ be as in the statement of Lemma \ref{LemmaCoupling}, and consider the probability measure
$$
P(G,X)dG=\frac{1}{Z_n}e^{-\Tr\left[G^*G\right]+b\Tr\left[GX+\left(GX\right)^*\right]}dG
$$
over rectangular complex matrices $G$. Here, $dG$ is the flat complex Lebesgue measure, and $Z_n$ is a normalising constant.
Then, the density of the squared singular values $y_1,\ldots,y_n$ of $Y=GX$ is
$$
P(y_1,\ldots,y_n)dy_1\ldots dy_n=\frac{\prod_{l=1}^n{{x_l}^{-1}e^{-b^2x_l}}}{b^{\frac{n(n-1)}{2}}}\frac{\det\left[e^{-\frac{y_j}{x_k}}\right]_{j,k=1}^n
\det\left[y_j^{\frac{k-1}{2}}I_{k-1}(2by_j)\right]_{j,k=1}^n}{n!\Delta_n(\{x_j\})}.
$$
\end{prop}
\begin{proof}By applying Lemma \ref{LemmaCoupling} (A) with $V=0$ the distribution follows. We only need to compute the following integral in order to determine the normalisation constant:
$$
I_n=\int\limits_{0}^{\infty}\ldots\int\limits_0^{\infty}\det\left[e^{-\frac{y_j}{x_k}}\right]_{j,k=1}^n
\det\left[y_j^{\frac{k-1}{2}}I_{k-1}\left(2by_j^{\frac{1}{2}}\right)\right]_{j,k=1}^ndy_1\ldots dy_n.
$$
Applying the Andr\'eief integral identity
\begin{equation}\label{Andre}
\int\cdots\int\det\left[\varphi_i\left(y_j\right)\right]_{i,j=1}^n
\det\left[\psi_i\left(y_j\right)\right]_{i,j=1}^ndy_1\cdots dy_n=n! \det\left[\int\varphi_i\left(y\right)\psi_j\left(y\right)dy\right]_{i,j=1}^n\ ,
\end{equation}
where the two sets of functions $\varphi_i$ and $\psi_i$ are assumed to be such that all integrals exist,
we find that
$$
I_n=n!\det\left[\int\limits_0^{\infty}e^{-\frac{y}{x_k}}y^{\frac{j-1}{2}}I_{j-1}\left(2by^{\frac{1}{2}}\right)dy\right]_{j,k=1}^n.
$$
The integral inside the determinant can be computed explicitly. The result is \cite[6.631.4]{Grad}
\begin{equation}
\label{I-int1}
\int\limits_0^{\infty}e^{-\frac{y}{x_k}}y^{\frac{j-1}{2}}I_{j-1}\left(2by^{\frac{1}{2}}\right)dy=(x_k)^jb^{j-1}e^{{b^2}x_k}.
\end{equation}
This gives
$$
I_n=b^{\frac{n(n-1)}{2}}\left(\prod\limits_{k=1}^nx_ke^{b^2x_k}\right)\Delta_n(\{x_j\}).
$$
The formula in the statement of Proposition \ref{TheoremDeformedDensity} follows immediately.
\end{proof}
Note that if $b=0$, then Proposition \ref{TheoremDeformedDensity} reduces to the following known result, cf.  Kuijlaars and Stivigny \cite[Lemma 2.2]{ArnoDries}
 and references therein. Let $G$ be a  complex  Ginibre matrix
of size $n\times (n+\nu)$, and let $X$  be a fixed complex matrix of size  $(n+\nu)\times n$ with nonzero squared singular values $x(X)=(x_1,\ldots,x_n)$.
Then the squared singular values  $y_1,\ldots,y_n$  of matrix $Y=GX$ have a joint probability density proportional to
$$
\frac{\Delta_n(\{y_j\})}{\Delta_n(\{x_j\})}\det\left[\frac{e^{-\frac{y_j}{x_k}}}{x_k}\right]_{j,k=1}^n.
$$
\section{Proof of Theorem \ref{TheoremGinibreCouplingDensity}}
Set $G=G_m$, of size $n\times (n+\nu_{m-1})$ (with $\nu_m=0)$, and $X=G_{m-1}\cdots G_1$ of size $(n+\nu_{m-1})\times n$ (with $\nu_0=0)$.  Denote by $y(GX)$ the vector $y^m=\left(y^m_1,\ldots,y^m_n\right)$, which is the vector of the squared singular values of the matrix
$Y_m=GX=G_mG_{m-1}\cdots G_1$. Note that $y^{m-1}=\left(y_1^{m-1},\ldots,y^{m-1}_n\right)$ is the vector of the squared values of $X$.
Let
$$
f:\underset{m\;\;\mbox{times}}{\underbrace{\left(\R_{>0}\right)^n\times\ldots\times\left(\R_{>0}\right)^n}}\longrightarrow\C
$$
be a continuous function   with compact support, and denote by $\E\left[f\right]$ the expectation
with respect to the  probability  measure defined by equation (\ref{MainProbabilityMeasure}). Using this notation we can write
\begin{equation}
\begin{split}
&\E\left[f\left(y^m;\ldots;y^1\right)\right]\\
&=\frac{1}{Z_n}\int\left(\int f(y(GX);y^{m-1};\ldots;y^1)
e^{-\Tr\left[G^*G\right]+b\Tr\left[GX
+\left(GX\right)^*\right]-\Tr V_m\left(\left(GX\right)^*\left(GX\right)\right)}
dG\right)\\
&\times e^{-\sum\limits_{l=1}^{m-1}\Tr\left[G_l^*G_l\right]
-\sum\limits_{l=1}^{m-1}\Tr V_l\left(\left(G_{l}\cdots G_1\right)^*\left(G_{l}\cdots G_1\right)\right)}\prod\limits_{l=1}^{m-1}dG_l.
\end{split}
\nonumber
\end{equation}
The application of Lemma \ref{LemmaCoupling} (more explicitly, of equation (\ref{CouplingFormula})) to the integral over $G$
in the equation just written above gives
\begin{equation}\label{MainProofExpectation1}
\begin{split}
&\E\left[f\left(y^m;\ldots;y^1\right)\right]\\
&=\int f_1\left(y^{m-1};\ldots;y^1\right)e^{-\sum\limits_{l=1}^{m-1}\Tr\left[G_l^*G_l\right]
-\sum\limits_{l=1}^{m-1}\Tr V_l\left(\left(G_{l}\cdots G_1\right)^*\left(G_{l}\cdots G_1\right)\right)}\prod\limits_{l=1}^{m-1}dG_l,
\end{split}
\end{equation}
where
\begin{equation}
\begin{split}
&f_1\left(y^{m-1};\ldots;y^1\right)\\
&=\frac{c_m}{Z_n}\int\limits_{\R_{>0}^n}
f\left(y^m; \ldots,y^{1}\right)
\det\left[\frac{e^{-\frac{y_j^m}{y_k^{m-1}}-V_m(y_j^m)}}{y_k^{m-1}}\right]_{j,k=1}^n
\frac{\det\left[\left(y_j^m\right)^{\frac{k-1}{2}}I_{k-1}
\left(2b\left(y_j^m\right)^{\frac{1}{2}}\right)\right]_{j,k=1}^n}{\Delta_n\left(\{y_i^{m-1}\}\right)}dy^m.
\end{split}
\nonumber
\end{equation}
Here, we denote by $dy^m=dy_1^m\cdots dy_n^m$ the integration measure over all squared singular values.

Next, let us apply Lemma \ref{LemmaCoupling}, equation (\ref{CouplingFormula1}) to equation (\ref{MainProofExpectation1}),
with $G=G_{m-1}$ of size $(n+\nu_{m-1})\times(n+\nu_{m-2})$, and $X=G_{m-2}\cdots G_1$ of size $(n+\nu_{m-2})\times n$, and integrate over $G_{m-1}$. The result can be written as
 \begin{equation}\label{MainProofExpectation2}
\begin{split}
&\E\left[f\left(y^m;\ldots;y^1\right)\right]\\
&=\int f_2\left(y^{m-2};\ldots;y^1\right)e^{-\sum\limits_{l=1}^{m-2}\Tr\left[G_l^*G_l\right]
-\sum\limits_{l=1}^{m-2}\Tr V_l\left(\left(G_{l}\cdots G_1\right)^*\left(G_{l}\cdots G_1\right)\right)}\prod\limits_{l=1}^{m-2}dG_l,
\end{split}
\end{equation}
where
\begin{equation}
\begin{split}
&f_2\left(y^{m-2};\ldots;y^1\right)\\
&=\frac{c_mc_{m-1}}{Z_n}\int\limits_{\R_{>0}^n}\int\limits_{\R_{>0}^n}
f\left(y^m; \ldots,y^{1}\right)
\det\left[\frac{e^{-\frac{y_j^m}{y_k^{m-1}}-V_m(y_j^m)}}{y_k^{m-1}}\right]_{j,k=1}^n
\\
&\times
\det\left[\frac{\left(y_j^{m-1}\right)^{\nu_{m-1}}e^{-\frac{y_j^{m-1}}{y_k^{m-2}}-V_{m-1}(y_j^{m-1})}}{\left(y_k^{m-2}\right)^{\nu_{m-1}+1}}\right]_{j,k=1}^n
\frac{\det\left[\left(y_j^m\right)^{\frac{k-1}{2}}I_{k-1}
\left(2b\left(y_j^m\right)^{\frac{1}{2}}\right)\right]_{j,k=1}^n}{\Delta_n\left(\{y_i^{m-2}\}\right)}
dy^{m-1}dy^m.
\end{split}
\nonumber
\end{equation}
Repeating this procedure $m-3$ times, we get
\begin{equation}\label{MainProofExpectation3}
\begin{split}
&\E\left[f\left(y^m;\ldots;y^1\right)\right]=\int f_{m-1}\left(y^1\right)e^{-\Tr\left[G_1^*G_1\right]
-\Tr V_1\left(G_1^*G_1\right)}dG_1,
\end{split}
\end{equation}
where
\begin{equation}\label{MainProofExpectation4}
\begin{split}
&f_{m-1}\left(y^1\right)\\
&=\frac{\prod_{l=2}^mc_l}{Z_n\Delta_n\left(\{y_i^1\}\right)}\int\limits_{\left(\R_{>0}^n\right)^{m-1}}
f\left(y^m; \ldots,y^{1}\right)
\det\left[\left(y_j^m\right)^{\frac{k-1}{2}}I_{k-1}\left(2b\left(y_j^m\right)^\frac{1}{2}\right)\right]_{j,k=1}^n\\
&\times
\prod\limits_{l=1}^{m-1}\det\left[\frac{\left(y_j^{l+1}\right)^{\nu_{l+1}}}{\left(y_k^l\right)^{\nu_{l+1}+1}}e^{-\frac{y_j^{l+1}}{y_k^l}
-V_{l+1}\left(y_j^{l+1}\right)}\right]_{j,k=1}^n
dy^2\ldots dy^m.
\end{split}
\end{equation}
The result of Theorem \ref{TheoremGinibreCouplingDensity} follows by application of formula (\ref{CouplingFormula2}) to equation
(\ref{MainProofExpectation3}), and by taking  into account equation (\ref{MainProofExpectation4}). Eq. \eqref{CouplingDensityFormula} is obtained by redistributing the factors with the potentials among the determinants, including the remaining Vandermonde determinant.
Finally the normalisation constant $Z_{n,m}$ in \eqref{Znmgen} is obtained by an $m$-fold application of the Andr\'eief formula \eqref{Andre}.
\qed

\section{Proof of Proposition \ref{PropositionJointDensityTotalProductMatrix}}
The integration of the density  (\ref{CouplingDensityFormula1}) of $\underline{y}=\left(y^m,y^{m-1},\ldots,y^1\right)$
over the sets of variables $y^{m-1}$, $\ldots$, $y^1$  gives Eq. (\ref{PExact}) by applying the Andr\'eief formula \eqref{Andre} $m-1$ times, with
\begin{equation}
\label{Psi-initial}
\begin{split}
\psi_k(y)=&\int\limits_{0}^{\infty}\ldots\int\limits_{0}^{\infty}
t^{\nu_1+k-1}\left(\frac{t_2}{t_1}\right)^{\nu_2}\ldots\left(\frac{t_{m-1}}{t_{m-2}}\right)^{\nu_{m-1}}
\left(\frac{y}{t_{m-1}}\right)^{\nu_m}\\
&\times e^{-t_1-\frac{t_2}{t_1}-\ldots-\frac{t_{m-1}}{t_{m-2}}-\frac{y}{t_{m-1}}-b^2t_{m-1}}
\frac{dt_1}{t_1}\frac{dt_2}{t_2}\ldots\frac{dt_{m-1}}{t_{m-1}}.
\end{split}
\end{equation}
We note that the integration over $t_1$, $t_2$, $\ldots$, $t_{m-2}$ in the formula
just written above results into the Meijer $G$-function (see e.g. \cite{AkemannIpsenKieburg}), namely
\begin{equation}
\label{Psi1}
\begin{split}
\psi_k(y)=&\int\limits_{0}^{\infty}
G^{m-1,0}_{0,m-1}\left(\begin{array}{cccc}
                          & - &  &  \\
                         \nu_{1}+k-1 & \nu_2 & \ldots & \nu_{m-1}
                       \end{array}
\biggl\vert t_{m-1}\right)e^{-\frac{y}{t_{m-1}}-b^2t_{m-1}}
\frac{dt_{m-1}}{t_{m-1}}.
\end{split}
\end{equation}
Using the contour integral representation for the Meijer $G$-function, cf. \cite{Luke},
\begin{equation}
\begin{split}
&G^{m-1,0}_{0,m-1}\left(\begin{array}{cccc}
                          & - &  &  \\
                         \nu_{1}+k-1 & \nu_2 & \ldots & \nu_{m-1}
                       \end{array}
\biggl\vert t_{m-1}\right)\\
&=\frac{1}{2\pi i}\int\limits_{c-i\infty}^{c+i\infty}
\Gamma\left(u+\nu_1+k-1\right)\Gamma\left(u+\nu_2\right)\ldots\Gamma\left(u+\nu_{m-1}\right)\left(t_{m-1}\right)^{-u}du,
\end{split}
\nonumber
\end{equation}
with $c>0$, and taking into account that \eqref{K-def2}
$$
\int\limits_0^{\infty}\left(t_{m-1}\right)^{-u-1}e^{-\frac{y}{t_{m-1}}-b^2t_{m-1}}dt_{m-1}=
2\left(\frac{y}{b^2}\right)^{-\frac{u}{2}}K_{u}\left(2b\sqrt{y}\right),
$$
we obtain equation (\ref{PSIK}).
The interchange of integrals can be justified with Fubini's Theorem.

The normalisation constant as given in \eqref{Znmb} can be obtained as follows. Applying the Andr\'eief formula \eqref{Andre} to \eqref{PExact} once, we need to compute
the
determinant of the following integral
\begin{equation}
a_{i,j}(b)=\int\limits_0^\infty y^{\frac{j-1}{2}} I_{j-1}\left( 2by^{\frac12}\right) \psi_i(y)  dy \ .
\end{equation}
Using the representation of the function $\psi_j(y)$ from \eqref{Psi-initial}, we observe that we can use the integral \eqref{I-int1} (with $x_k$ replaced by $t_{m-1}$)
to obtain
\begin{equation}
\label{aijb}
a_{i,j}(b)=b^{j-1}\int\limits_{0}^{\infty}\ldots\int\limits_{0}^{\infty}
t^{\nu_1+i-1}\left(\frac{t_2}{t_1}\right)^{\nu_2}\ldots\left(\frac{t_{m-1}}{t_{m-2}}\right)^{\nu_{m-1}}
(t_{m-1})^{j}
e^{-t_1-\frac{t_2}{t_1}-\ldots-\frac{t_{m-1}}{t_{m-2}}
}
\frac{dt_1}{t_1}
\ldots\frac{dt_{m-1}}{t_{m-1}}.
\end{equation}
The powers in $b$ can be taken out of the determinant of $a_{i,j}(b)$, and the remaining integral is the same as in the normalisation of $m-1$ product of independent Ginibre matrices. For these the determinant has been computed in \cite{AkemannIpsenKieburg} and we thus obtain for
\begin{equation}
\det\left[ a_{i,j}(b)\right]_{i,j=1}^n=b^{\frac{n(n-1)}{2}}\prod_{j=1}^n\prod_{l=1}^{m-1}\Gamma(j+\nu_l)\ .
\end{equation}
Together with the $n!$ from the Andr\'eief formula \eqref{Andre} we obtain \eqref{Znmb}.
\qed

\section{Measures given by products of determinants, Eynard-Mehta Theorem, and proof of Theorem \ref{TheoremMostGerneralCorrelationKernel}}\label{SECTIONEynardMehta}
The aim of this section is to prove Theorem \ref{TheoremMostGerneralCorrelationKernel}.
Recall that Theorem \ref{TheoremMostGerneralCorrelationKernel} states that the product matrix process associated with probability measure
(\ref{MainProbabilityMeasure}), defined in Section \ref{SectionExactResults1}, is a multi-level determinantal process
living on $\left\{1,\ldots,m\right\}\times\R_{>0}$. Moreover, Theorem \ref{TheoremMostGerneralCorrelationKernel} gives a formula for the
relevant correlation kernel.

The starting point of the proof of Theorem 2.4 is the fact that the density of the product matrix process
under considerations is given by a product of determinants, see Theorem \ref{TheoremGinibreCouplingDensity}.
This enables us to apply the Eynard-Mehta Theorem to the product matrix process.

Let us first recall the formulation of the Eynard-Mehta Theorem. Here we follow the elegant presentation
of the Eynard-Mehta Theorem in Johansson \cite{Johansson}\footnote{For other presentations of the Eynard-Mehta
Theorem, and for different proofs we refer the reader to Borodin \cite{Borodin}, Tracy and Widom \cite{TracyWidom}.}.

Let $n, m\geq 1$ be two fixed natural numbers, and let $\X_0$, $\X_{m+1}$ be two given sets. Let $\X$ be a complete separable
metric space, and consider a probability measure on $(\X^n)^m$ given by
\begin{equation}\label{ProductDeterminantsMeasure}
\begin{split}
p_{n,m}(\underline{x})d\mu(\underline{x})&=\frac{1}{Z_{N,m}}\det\left[\phi_{0,1}(x_i^0,x_j^1)\right]_{i,j=1}^n\det\left[\phi_{m,m+1}(x_i^m,x_j^{m+1})\right]_{i,j=1}^n\\
&
\times\prod\limits_{r=1}^{m-1}\det\left[\phi_{r,r+1}(x_i^r,x_j^{r+1})\right]_{i,j=1}^nd\mu(\underline{x}).
\end{split}
\end{equation}
In the formula just written above $Z_{N,m}$ is the normalisation constant, the functions $\phi_{r,r+1}: \X\times \X\rightarrow \C$, $r=1,\ldots,m-1$
are given \textit{intermediate one-step  transition functions},  $\phi_{0,1}: \X_0\times \X\rightarrow \C$ is a given \textit{initial one-step transition function},
and $\phi_{m,m+1}: \X\times \X_{m+1}\rightarrow \C$ is a given \textit{final one-step transition function}. Also,
$$
\underline{x}=\left(x^1,\ldots,x^m\right)\in\left(X^n\right)^m; \;\; x^r=\left(x^r_1,\ldots,x^r_n\right), r=1,\ldots, m,
$$
the vectors
$$
x^0=(x^0_1,\ldots,x^0_n)\in \X_0^n,\;\; x^{m+1}=(x^{m+1}_1,\ldots,x^{m+1}_n)\in \X_{m+1}^n,
$$
are fixed initial and final vectors,
and
$$
d\mu(\underline{x})=\prod\limits_{r=1}^m\prod\limits_{j=1}^nd\mu(x_j^r).
$$
Here, $\mu$ is a given Borel measure on $\X$.
Given two transition functions $\phi$ and $\psi$ set
$$
\phi\ast\psi(x,y)=\int_{\X}\phi(x,t)\psi(t,y)d\mu(t).
$$
\begin{prop} Consider the probability measure defined by equation (\ref{ProductDeterminantsMeasure}). The distribution of the vector
$x^m=\left(x_1^m,\ldots,x_n^m\right)$ is given by
\begin{equation}\label{PmGeneral}
\begin{split}
&p(x_1^m,\ldots,x_n^m)d\mu(x_1^m)\ldots d\mu(x_n^m)\\
&=\frac{1}{Z_{n,m}'}\det\left[\phi_{0,m}(x_i^0,x_j^m)\right]_{i,j=1}^n
\det\left[\phi_{m,m+1}(x_i^m,x_j^{m+1})\right]_{i,j=1}^nd\mu(x_1^m)\ldots d\mu(x_n^m),
\end{split}
\nonumber
\end{equation}
where
$$
\phi_{0,m}(x,y)=\phi_{0,1}\ast\ldots\ast\phi_{m-1,m}(x,y),
$$
 for $m>1$, and where the normalisation constant is given by
$$
Z_{n,m}'=n!\det\left[\phi_{0,m+1}(x_i^0,x_j^{m+1})\right]_{i,j=1}^n.
$$
\end{prop}
\begin{proof}
The density of $\left(x_1^m,\ldots,x_n^m\right)$ can be obtained by subsequent integration of the measure $p_{n,m}(\underline{x})d\mu(\underline{x})$
over $x^1$, $\ldots$, $x^{m-1}$, and by application of the Andr\'eief integral identity \eqref{Andre}.
\end{proof}

Let us define the following correlation functions
for the process defined by probability measure (\ref{ProductDeterminantsMeasure}):
\begin{equation}
\label{rhok-def}
\varrho_{k_1,\ldots,k_m}\left(x_1^1,\ldots,x_{k_1}^1;\ldots;x_1^m,\ldots,x_{k_m}^m\right)
= \prod_{j=1}^m\frac{n!}{(n-k_j)!}
\int_{\X}\cdots\int_{\X}p_{n,m}(\underline{x})\prod\limits_{r=1}^m\prod\limits_{j=k_r+1}^nd\mu(x_j^r).
\end{equation}
The following statement determines these for the point process \eqref{ProductDeterminantsMeasure} and is often referred as the Eynard-Mehta Theorem \cite{EynardMehta}.
\begin{thm}\label{TheoremEynardMehta} The probability measure $p_{n,m}(\underline{x})d\mu(\underline{x})$
given by equation (\ref{ProductDeterminantsMeasure})  defines a determinantal point process on
$\left\{1,\ldots,m\right\}\times \X$. The correlation kernel of this determinantal point process, $K_{n,m}(r,x;s,y)$
(where $r,s\in\left\{1,\ldots,m\right\}$, and $x, y\in \X)$, is given by the formula
\begin{equation}
K_{n,m}(r,x;s,y)=-\phi_{r,s}(x,y)+\sum\limits_{i,j=1}^n\phi_{r,m+1}(x,x_i^{m+1})\left(A^{-1}\right)_{i,j}\phi_{0,s}(x_j^0,y).
\end{equation}
The additional transition functions $\phi_{r,s}$ with $s\neq r+1$, and the matrix $A=(a_{i,j})$, with $i,j=1,\ldots,n$, are defined as follows in terms of
the one-step transition functions $\phi_{r,r+1}$, with $r=0,1,\ldots,m$, of point process \eqref{ProductDeterminantsMeasure}:
\begin{equation}
\label{trans-def}
\phi_{r,s}(x,y)=\left\{
                  \begin{array}{ll}
                    \left(\phi_{r,r+1}\ast\ldots\ast\phi_{s-1,s}\right)(x,y), & 0\leq r<s\leq m+1, \\
                    0, & r\geq s,
                  \end{array}
                \right.
\end{equation}
and
\begin{equation}
a_{i,j}=\phi_{0,m+1}(x_i^0,x_j^{m+1}).
\end{equation}
The correlation functions defined in \eqref{rhok-def}
can be written as determinants of block matrices, namely
\begin{equation}
\begin{split}
&\varrho_{k_1,\ldots,k_m}\left(x_1^1,\ldots,x_{k_1}^1;\ldots;x_1^m,\ldots,x_{k_m}^m\right)\\
&=\det\left[\begin{array}{ccc}
             \left(K_{n,m}(1,x_i^1;1,x_j^1)\right)_{i=1,\ldots,k_1}^{j=1,\ldots,k_1}
              & \ldots & \left(K_{n,m}(1,x_i^1;m,x_j^m)\right)_{i=1,\ldots,k_1}^{j=1,\ldots,k_m} \\
              \vdots &  &  \\
 \left(K_{n,m}(m,x_i^m;1,x_j^1)\right)_{i=1,\ldots,k_m}^{j=1,\ldots,k_1}
  & \ldots & \left(K_{n,m}(m,x_i^m;m,x_j^m)\right)_{i=1,\ldots,k_m}^{j=1,\ldots,k_m}
            \end{array}
\right],
\end{split}
\label{rhoK}
\end{equation}
where $1\leq k_1,\ldots,k_m\leq n$, and for $1\leq l,p\leq m$
$$
\left(K_{n,m}(l,x_i^l;p,x_j^p)\right)_{i=1,\ldots,k_l}^{j=1,\ldots,k_p}
=\left(\begin{array}{ccc}
         K_{n,m}(l,x_1^l;p,x_1^p) & \ldots & K_{n,m}(l,x_1^l;p,x_{k_p}^p)\\
         \vdots &  &  \\
         K_{n,m}(l,x_{k_l}^l;p,x_1^p) & \ldots & K_{n,m}(l,x_{k_l}^l;p,x_{k_p}^p)
       \end{array}
\right).
$$
\end{thm}

\begin{Remarks}
In what follows  the functions
$$
\phi_{0,s}(i,y),\;\; 2\leq s\leq m,
$$
will be called \textit{initial transition functions}, and the functions
$$
\phi_{r,m+1}(x,j),\;\; 1\leq r\leq m-1,
$$
will be called  \textit{final transition functions}. In addition, the functions of the form
$$
\phi_{r,s}(x,y),\;\; 1\leq r\leq m-2,\;\; r+2\leq s\leq m,
$$
will be called \textit{intermediate transition functions.} Finally, the function
$$
\phi_{0,m+1}(i,j)
$$
will be called \textit{the total transition function}.
\end{Remarks}
In order to prove Theorem \ref{TheoremMostGerneralCorrelationKernel} we need to rewrite
the density of the product matrix process obtained in  Theorem \ref{TheoremGinibreCouplingDensity}
as in the formulation of the Eynard-Mehta Theorem, see equation (\ref{ProductDeterminantsMeasure}),
and to obtain explicit expressions for the relevant transition functions. This is done below.\\
$\bullet$ \textbf{One-step transition functions.} Recall that $\nu_0=\nu_m=0$.
In our situation $\X_0=\left\{1,\ldots,n\right\}$, $\X_{m+1}=\left\{1,\ldots,n\right\}$, $\X=\R_{>0}$, and $d\mu$ is the Lebesgue
measure on $\R_{>0}$. The initial given one-step transition function is defined by
\begin{equation}
\label{phi01}
\phi_{0,1}: \left\{1,\ldots,n\right\}\times\R_{>0}\rightarrow\R_{>0};\;\;\; \phi_{0,1}(i,x)=x^{\nu_1+i-1}e^{-x}.
\end{equation}
The final given one-step transition function is defined by
\begin{equation}
\label{phimm+1}
\phi_{m,m+1}: \R_{>0}\times\left\{1,\ldots,n\right\}\rightarrow\R_{>0};\;\;\;\phi_{m,m+1}(x,k)=x^{\frac{k-1}{2}}I_{k-1}\left(2bx^{\frac{1}{2}}\right)e^{-V_m(x)}.
\end{equation}
In addition, the intermediate given one-step transition functions
$$
\phi_{r,r+1}: \R_{>0}\times\R_{>0}\rightarrow\R_{>0},\;\;\; r=1,\ldots,m-1,
$$
are defined by
\begin{equation}
\label{phirr+1}
\phi_{r,r+1}(x,y)=\left(\frac{y}{x}\right)^{\nu_{r+1}}\frac{e^{-\frac{y}{x}-V_{r}(x)}}{x},\;\;\;
 r=1,\ldots,m-1.
\end{equation}
$\bullet$ \textbf{Initial transition functions.}
The initial transition functions, $\phi_{0,s}(i,y)$, with $2\leq s\leq m$, can be written as
\begin{equation}
\begin{split}
\phi_{0,s}(i,y)
&=\phi_{0,1}\ast\phi_{1,2}\ast\ldots\ast\phi_{s-2,s-1}\ast\phi_{s-1,s}(i,y)\\
&=\int\limits_0^{\infty}\ldots\int\limits_0^{\infty}
\phi_{0,1}(i,t_1)\phi_{1,2}(t_1,t_2)\ldots\phi_{s-2,s-1}(t_{s-2},t_{s-1})\phi_{s-1,s}(t_{s-1},y)dt_1
\ldots
dt_{s-1},
\end{split}
\end{equation}
spelling out the convolution in the last line.
Inserting the corresponding expressions for the one-step transition functions \eqref{phi01} and \eqref{phirr+1} we obtain
\begin{equation}
\begin{split}
\phi_{0,s}(i,y)
=&\int\limits_0^{\infty}\ldots\int\limits_0^{\infty}
t_1^{\nu_1+i-1}\left(\frac{t_2}{t_1}\right)^{\nu_2}\ldots\left(\frac{t_{s-1}}{t_{s-2}}\right)^{\nu_{s-1}}
\left(\frac{y}{t_{s-1}}\right)^{\nu_s}
\\
&\times e^{-t_1-\frac{t_2}{t_1}-\ldots-\frac{t_{s-1}}{t_{s-2}}-\frac{y}{t_{s-1}}-
V_1\left(t_1\right)
-\ldots
-V_{s-1}\left(t_{s-1}\right)}
\frac{dt_1}{t_1}
\ldots
\frac{dt_{s-1}}{t_{s-1}},
\end{split}
\end{equation}
where $2\leq s\leq m$ (for $s=1$ see \eqref{phi01}).\\
$\bullet$ \textbf{Final transition functions.}
The final transition functions $\phi_{r,m+1}(x,j)$, with
\newline
$1\leq r\leq m-1$, can be written in terms of the one-step transition functions as
\begin{equation}
\begin{split}
&\phi_{r,m+1}(x,j)\\
&=\phi_{r,r+1}\ast\phi_{r+1,r+2}\ast\ldots\ast\phi_{m-1,m}\ast\phi_{m,m+1}(x,j)\\
&=\int\limits_0^{\infty}\ldots\int\limits_0^{\infty}
\phi_{r,r+1}(x,t_{r+1})\phi_{r+1,r+2}(t_{r+1},t_{r+2})\ldots\phi_{m-1,m}(t_{m-1},t_{m})\phi_{m,m+1}(t_{m},j)dt_{r+1}
\ldots dt_{m}.
\end{split}
\end{equation}
Inserting the explicit formulae for the one-step transition functions \eqref{phirr+1} and \eqref{phimm+1} we find
\begin{equation}
\label{phirm+1}
\begin{split}
\phi_{r,m+1}(x,j)
=&\frac{e^{-V_r(x)}}{x}\int\limits_0^{\infty}\ldots\int\limits_0^{\infty}
\left(\frac{t_{r+1}}{x}\right)^{\nu_{r+1}}
\left(\frac{t_{r+2}}{t_{r+1}}\right)^{\nu_{r+2}}
\ldots
\left(\frac{t_m}{t_{m-1}}\right)^{\nu_{m}}
\left(t_m\right)^{\frac{j-1}{2}}I_{j-1}\left(2b\left(t_m\right)^{\frac{1}{2}}\right)
\\
&\times e^{-\frac{t_{r+1}}{x}-\frac{t_{r+2}}{t_{r+1}}-\ldots
-\frac{t_m}{t_{m-1}}-V_{r+1}\left(t_{r+1}\right)
\ldots
-V_{m}\left(t_{m}\right)}
\frac{dt_{r+1}}{t_{r+1}}
\ldots \frac{dt_{m-1}}{t_{m-1}}dt_m,
\end{split}
\end{equation}
where $1\leq r\leq m-1$ (for $r=m$ see \eqref{phimm+1}).
\newpage
$\bullet$ \textbf{Intermediate transition functions.}
Recall that the intermediate transition functions $\phi_{r,s}(x,y)$, $1\leq r\leq m-2$, $r+2\leq s\leq m$ are defined by
\begin{equation}
\begin{split}
\phi_{r,s}(x,y)
&=\phi_{r,r+1}\ast\phi_{r+1,r+2}\ast\ldots\ast\phi_{s-1,s}(x,y)\\
&=\int\limits_0^{\infty}\ldots\int\limits_0^{\infty}
\phi_{r,r+1}(x,t_{r+1})\phi_{r+1,r+2}(t_{r+1},t_{r+2})\ldots\phi_{s-1,s}(t_{s-1},y)dt_{r+1}
\ldots dt_{s-1},
\end{split}
\end{equation}
where $1\leq r\leq m-2$, $r+2\leq s\leq m$. Using the explicit formulae for the intermediate one-step transition functions \eqref{phirr+1} we get
\begin{equation}
\begin{split}
\phi_{r,s}(x,y)
=&\frac{e^{-V_r(x)}}{x}\int\limits_0^{\infty}\ldots\int\limits_0^{\infty}
\left(\frac{t_{r+1}}{x}\right)^{\nu_{r+1}}
\left(\frac{t_{r+2}}{t_{r+1}}\right)^{\nu_{r+2}}
\ldots
\left(\frac{y}{t_{s-1}}\right)^{\nu_{s}}
\\
&\times e^{-\frac{t_{r+1}}{x}-\frac{t_{r+2}}{t_{r+1}}-\ldots-\frac{t_{s-1}}{t_{s-2}}-\frac{y}{t_{s-1}}-V_{r+1}\left(t_{r+1}\right)
\ldots-V_{s-1}\left(t_{s-1}\right)}
\frac{dt_{r+1}}{t_{r+1}}
\ldots
\frac{dt_{s-1}}{t_{s-1}},
\end{split}
\end{equation}
where $1\leq r\leq m-2$, $r+2\leq s\leq m$ (for $s=r+1$ see \eqref{phirr+1}).\\
$\bullet$ \textbf{Total transition function.}
The total transition function
$
\phi_{0,m+1}(i,j)=a_{i,j}
$
that constitutes matrix $A$ can be written as the convolution  of the one-step initial transition function $\phi_{0,1}$ from \eqref{phi01} and the final transition function
$\phi_{1,m+1}$ from  \eqref{phirm+1} for $r=1$:
$$
\phi_{0,m+1}(i,j)=\phi_{0,1}\ast\phi_{1,m+1}(i,j)=\int\limits_0^{\infty}
\phi_{0,1}(i,t_{1})\phi_{1,m+1}(t_{1},j)dt_{1}.
$$
This gives
\begin{equation}
\begin{split}
\phi_{0,m+1}(i,j)
&=\int\limits_0^{\infty}\ldots\int\limits_0^{\infty}
t_1^{\nu_1+i-1}\left(\frac{t_2}{t_1}\right)^{\nu_2}\ldots
\left(\frac{t_m}{t_{m-1}}\right)^{\nu_m}
\left(t_m\right)^{\frac{j-1}{2}}I_{j-1}\left(2b\left(t_m\right)^{\frac{1}{2}}\right)
\\
&\times e^{-t_1-\frac{t_2}{t_1}-\ldots
-\frac{t_m}{t_{m-1}}
-V_1\left(t_1\right)
-\ldots
-V_{m}\left(t_{m}\right)}
\frac{dt_1}{t_1}
\ldots \frac{dt_{m-1}}{t_{m-1}} dt_m.
\end{split}
\end{equation}
Once all the transition functions are written explicitly, Theorem \ref{TheoremMostGerneralCorrelationKernel} follows
immediately from Theorem \ref{TheoremEynardMehta}.
\qed
\section{Double contour integral representation for the correlation kernel}
In this section we consider the $m$-matrix model defined by probability measure
(\ref{MainProbabilityMeasure1}). The correlation kernel of the product matrix process associated with this model
is denoted by $K_{n,m}\left(r,x;s,y;b\right)$. Our aim is to derive a double contour integral
representation for  $K_{n,m}\left(r,x;s,y;b\right)$, and to prove  Theorem \ref{TheoremDoubleIntegralRepresentationExactKernel}.
Note that the multi-matrix model defined by probability measure (\ref{MainProbabilityMeasure}) turns into that
defined by probability measure (\ref{MainProbabilityMeasure1}) if
\begin{equation}\label{Potentials}
V_1(t)=\ldots=V_{m-2}(t)=0,\;\; V_{m-1}(t)=b^2t,\;\; V_m(t)=0.
\end{equation}
Therefore, in order to derive a contour integral representation for the correlation kernel
$K_{n,m}\left(r,x;s,y;b\right)$ \eqref{KnmV} we can exploit the formulae obtained in Theorem \ref{TheoremMostGerneralCorrelationKernel}
with the potential functions specified by (\ref{Potentials}).
\begin{prop}\label{PropositionKV}
For the specific case $V_1(t)=\ldots= V_{m-2}(t)=0$, 
$V_{m-1}(t)=b^2t$, $V_m(t)=0$ the correlation kernel  $K_{n,m}(r,x;s,y;b)$
can be written  as
\begin{equation}\label{CorrelationKernelGeneralFormulaGeneral}
K_{n,m}(r,x;s,y;b)=-\phi_{r,s}(x,y;b)+\sum\limits_{i,j=1}^n\phi_{r,m+1}(x,i)\left(A^{-1}\right)_{i,j}\phi_{0,s}(j,y).
\end{equation}
Here, the three sets of functions are obtained as follows.

\textit{(i)}  The functions $\phi_{r,s}(x,y;b)$ are given by
\begin{equation}
\begin{split}
&\phi_{r,s}(x,y;b)\\
&=\left\{
                  \begin{array}{lll}
                  \frac{e^{-\frac{y}{x}-b^2x}}{x}, & r=m-1, s=m,\\
                    \frac{1}{x}
G^{s-r,0}_{0,s-r}\left(\begin{array}{cccc}
                - \\
               \nu_{r+1},\ldots,\nu_{s}
             \end{array}\biggl|\frac{y}{x}\right), &  1\leq r<s\leq m-1,\\
                    \frac{1}{x}
\int\limits_{0}^{\infty}
G^{m-r-1,0}_{0,m-r-1}\left(\begin{array}{cccc}
                - \\
               \nu_{r+1},\ldots,\nu_{m-1}
             \end{array}\biggl|\frac{t}{x}\right)
             e^{-\frac{y}{t}-b^2t}\frac{dt}{t}, & 1\leq r\leq m-2, s=m,\\
0, & 1\leq s\leq r\leq m.\\
                  \end{array}
                \right.
\end{split}
\nonumber
\end{equation}
\textit{(ii)} The functions $\phi_{r,m+1}(x,i)$ are given by
\begin{equation}
\phi_{r,m+1}(x,j)=
\left\{
  \begin{array}{ll}
    b^{j-1}x^{j-1}\Gamma(j+\nu_{r+1})\ldots\Gamma(j+\nu_{m-1}), & r\in\left\{1,\ldots,m-2\right\},\\
    b^{j-1}x^{j-1}, & r=m-1, \\
    x^{\frac{j-1}{2}}I_{j-1}\left(2bx^{\frac{1}{2}}\right), & r=m.
  \end{array}
\right.
\nonumber
\end{equation}
\textit{(iii)} The functions
$\phi_{0,s}(i,y)$ are given by
\begin{equation}
\label{phi0siy}
\phi_{0,s}(i,y)=\left\{
  \begin{array}{ll}
G^{s,0}_{0,s}\left(\begin{array}{cccc}
                - \\
              \nu_1+i-1, \nu_{2},\ldots,\nu_{s}
             \end{array}\biggl|y\right)& s\in\left\{1,\ldots,m-1\right\}\\
             \int\limits_{0}^{\infty}
G^{m-1,0}_{0,m-1}\left(\begin{array}{cccc}
                - \\
              \nu_1+i-1, \nu_{2},\ldots,\nu_{m-1}
             \end{array}\biggl|t\right)
             e^{-\frac{y}{t}-b^2t}\frac{dt}{t} & s=m.\\
\end{array}
  \right.
\end{equation}

Finally, for the matrix $A=\left(a_{i,j}\right)_{i,j=1}^n$, we have
\begin{equation}\label{aij}
a_{i,j}=\phi_{0,m+1}(i,j)=b^{j-1}\Gamma(i+j-1+\nu_{1})\Gamma(j+\nu_2)\ldots\Gamma(j+\nu_{m-1}).
\end{equation}
\end{prop}
\begin{proof}
The formulae for the transition functions stated in Proposition \ref{PropositionKV}
can be obtained by straightforward calculations starting from the formulae obtained in Theorem \ref{TheoremMostGerneralCorrelationKernel}.
In the calculations we have exploited the following integral representation
of the Meijer $G$-function, cf. \cite{AkemannIpsenKieburg},
\begin{equation}\label{GElementaryRepresentation}
\begin{split}
&G^{m,0}_{0,m}\left(\begin{array}{ccc}
                - \\
               b_1,\ldots,b_m
             \end{array}\biggl|\frac{x_m}{x_0}\right)\\
&=\int\limits_0^{\infty}\ldots\int\limits_0^{\infty}\left(\frac{x_1}{x_0}\right)^{b_1}\left(\frac{x_2}{x_1}\right)^{b_2}\ldots\left(\frac{x_m}{x_{m-1}}\right)^{b_m}
e^{-\frac{x_1}{x_0}-\frac{x_2}{x_1}-\ldots-\frac{x_m}{x_{m-1}}}\frac{dx_1}{x_1}\ldots\frac{dx_{m-1}}{x_{m-1}},
\end{split}
\end{equation}
the fact
that the Mellin transform of a Meijer $G$-function is given by
\begin{equation}
\begin{split}
\int\limits_0^{\infty}t^{u-1}G^{m,n}_{p,q}\left(\begin{array}{ccc}
                a_1,\ldots, a_p \\
               b_1,\ldots, b_q
             \end{array}\biggl|tz\right)dt
=z^{-u}\frac{\prod_{i=1}^m\Gamma(b_i+u)\prod_{i=1}^n\Gamma(1-a_i-u)}{\prod_{i=n+1}^p\Gamma(a_i+u)\prod_{i=m+1}^q\Gamma(1-b_i-u)},
\end{split}
\end{equation}
see Luke \cite{Luke}, and the  integral \eqref{I-int1} involving the modified Bessel functions of the first kind
\end{proof}
In Proposition \ref{PropositionKV} we have found the matrix entries of $A$ explicitly.
This enables us to derive a formula for the inverse of $A$.
\begin{prop}\label{PropositionInverse} Let $C=\left(c_{j,k}\right)_{j,k=1}^n$ be the inverse of $A=\left(a_{i,j}\right)_{i,j=1}^n$,
where $a_{i,j}$ is given by equation (\ref{aij}).
Thus the matrix elements $\left(c_{j,k}\right)_{j,k=1}^n$ are defined by the relation
$$
\sum\limits_{j=1}^na_{i,j}c_{j,k}=\delta_{i,k},\;\;\; i,k\in\left\{1,\ldots,n\right\}.
$$
We have\footnote{Recall that the Pochhammer symbol, $(a)_l$, is defined by
$
(a)_l=a(a+1)\ldots(a+l-1).
$}
\begin{equation}
\label{cjk}
\begin{split}
c_{j,k}=&\frac{1}{b^{j-1}\Gamma\left(j+\nu_2\right)\ldots\Gamma\left(j+\nu_{m-1}\right)}\\
&\times\sum\limits_{p=0}^{n-1}\frac{\Gamma\left(\nu_1+p+1\right)}{\Gamma\left(\nu_1+j\right)\Gamma\left(\nu_1+k\right)p!}
\frac{\left(-p\right)_{j-1}\left(-p\right)_{k-1}}{\left(\nu_1+1\right)_{j-1}\left(\nu_1+1\right)_{k-1}}\frac{1}{\left(j-1\right)!\left(k-1\right)!}.
\end{split}
\end{equation}
\end{prop}
\begin{proof} Clearly, the matrix elements $\left(c_{j,k}\right)_{j,k=1}^n$ can be written as
$$
c_{j,k}=\frac{C_{j,k}}{b^{j-1}\Gamma\left(j+\nu_2\right)\ldots\Gamma\left(j+\nu_{m-1}\right)},
$$
where $C_{j,k}$ is defined by
$$
\sum\limits_{j=1}^n\Gamma\left(i+j-1+\nu_1\right)C_{j,k}=\delta_{i,k},\;\;\;i,k\in\left\{1,\ldots,n\right\}.
$$
Next we use the fact that the inverse $\left(\alpha_{k,l}\right)_{k,l=0}^{N-1}$
of the Hankel matrix
$$
\left(h_{k+l}\right)_{k,l=0}^{N-1},\;\; h_k=\Gamma(k+\nu+1),
$$
is given by
$$
\alpha_{k,l}=\sum\limits_{p=0}^{N-1} \frac{\Gamma(\nu+p+1)(-p)_{k}(-p)_{l}}{p!\Gamma(\nu_1+k+1)\Gamma(\nu_1+l+1)k!l!},
$$
see, for example, Akemann and Strahov \cite[Section 5]{AkemannStrahov}. Equation \eqref{cjk} follows.
\end{proof}
The fact that the inverse of the matrix $A$ can be written explicitly
allows us to spell out the correlation kernel  $K_{n,m}(r,x;s,y;b)$ of Proposition \ref{PropositionKV}. We start from  equation
\eqref{CorrelationKernelGeneralFormulaGeneral}, use Proposition \ref{PropositionInverse},  and obtain the formula
\begin{equation}\label{FirstFormulaForK}
\begin{split}
K_{n,m}\left(r,x;s,y;b\right)=-\phi_{r,s}\left(x,y;b\right)+\sum\limits_{p=0}^{n-1}\frac{\Gamma\left(\nu_1+p+1\right)}{\Gamma^2(\nu_1+1)p!}
P_{r,p}(x)Q_{s,p}(y),
\end{split}
\end{equation}
where the transition functions $\phi_{r,s}\left(x,y\right)$ are defined in   Proposition \ref{PropositionKV} (i),
\begin{equation}
P_{r,p}(x)=\sum\limits_{i=0}^p\phi_{r,m+1}\left(x,i+1\right)\frac{(-p)_i}{\Gamma\left(i+\nu_2+1\right)\ldots\Gamma\left(i+\nu_{m-1}+1\right)
(\nu_1+1)_ii!b^{i}}
\end{equation}
and
\begin{equation}\label{QSP}
Q_{s,p}(y)=\sum\limits_{j=0}^p\phi_{0,s}\left(j+1,y\right)\frac{\left(-p\right)_j}{\left(\nu_1+1\right)_jj!}.
\end{equation}
Note that the Pochhammer symbol truncates both sums in $P_{r,p}(x)$ and $Q_{s,p}(y)$ that initially ran up to $n-1$.
Proposition \ref{PropositionKV} gives the transition functions $\phi_{r,m+1}\left(x,i+1\right)$ and $\phi_{0,s}\left(j+1,y\right)$
explicitly. This enables us to obtain different, useful formulae for the functions $P_{r,p}(x)$ and $Q_{s,p}(y)$, involved
 in expression (\ref{FirstFormulaForK}) for the correlation kernel $K_{n,m}\left(r,x;s,y;b\right)$.\\
$\bullet$ \textbf{The functions} $P_{r,p}(x)$.
Let us
derive a contour integration formula
for the functions $P_{r,p}(x)$. If $r\in\left\{1,\ldots,m-1\right\}$, we have
\begin{equation}
\begin{split}
P_{r,p}(x)&=\Gamma\left(\nu_1+1\right)
\sum\limits_{i=0}^p\frac{x^i\left(-p\right)_i}{i!\Gamma\left(\nu_1+i+1\right)
\ldots\Gamma\left(\nu_{r}+i+1\right)}\\
&=(-1)^p\Gamma\left(\nu_1+1\right)p!\sum\limits_{i=0}^p
\frac{(-1)^{p-i}x^i}{i!(p-i)!\Gamma\left(\nu_1+i+1\right)
\ldots\Gamma\left(\nu_r+i+1\right)}.
\end{split}
\end{equation}
For $r=m$, the formula for the function $P_{m,p}(x)$ is
\begin{equation}
\label{Pmn}
\begin{split}
P_{m,p}(x)&=\Gamma\left(\nu_1+1\right)
\sum\limits_{i=0}^p\frac{x^{\frac{i}{2}}\left(-p\right)_i}{\Gamma\left(\nu_1+i+1\right)
\ldots\Gamma\left(\nu_{m-1}+i+1\right)}\frac{I_i\left(2bx^{\frac{1}{2}}\right)}{b^ii!}\\
&=(-1)^p\Gamma\left(\nu_1+1\right)p!\sum\limits_{i=0}^p
\frac{(-1)^{p-i}x^{\frac{i}{2}}}{(p-i)!\Gamma\left(\nu_1+i+1\right)
\ldots\Gamma\left(\nu_{m-1}+i+1\right)}\frac{I_i\left(2bx^{\frac{1}{2}}\right)}{b^ii!}.
\end{split}
\end{equation}
Using the expressions for the functions $P_{r,p}(x)$ just written above together with
the Residue Theorem we immediately
obtain a contour integral representation for these functions. Namely,
for $r\in\left\{1,\ldots,m-1\right\}$ we find
\begin{equation}\label{PINTEGRAL1}
P_{r,p}(x)=\frac{(-1)^p\Gamma\left(\nu_1+1\right)p!}{2\pi i}
\oint\limits_{\Sigma_p}\frac{\Gamma\left(t-p\right)}{\prod_{j=0}^r\Gamma\left(t+\nu_j+1\right)}
x^tdt,
\end{equation}
and for $r=m$ we find
\begin{equation}\label{PINTEGRAL2}
P_{m,p}(x)=\frac{(-1)^p\Gamma\left(\nu_1+1\right)p!}{2\pi i}
\oint\limits_{\Sigma_p}\frac{\Gamma\left(t-p\right)x^{\frac{t}{2}}}{\prod_{j=0}^{m-1}\Gamma\left(t+\nu_j+1\right)}
\frac{I_t\left(2bx^{\frac{1}{2}}\right)}{b^t}dt.
\end{equation}
In the formulae just written above $\Sigma_p$ denotes a closed contour encircling 
$0,1,\ldots,p$ in positive direction.\\
$\bullet$ \textbf{The functions $Q_{s,p}(y)$.}
For $s\in\left\{1,\ldots,m-1\right\}$ the formula for the functions $Q_{s,p}(y)$ reads
\begin{equation}
Q_{s,p}(y)=
\sum\limits_{j=0}^p\frac{(-p)_jG^{s,0}_{0,s}\left(\begin{array}{cccc}
                - \\
               \nu_1+j,\nu_2\ldots,\nu_s
             \end{array}\biggl|y\right)}{\left(\nu_1+1\right)_jj!}.
\end{equation}
The contour integral representation of the Meijer $G$-function above is
\begin{equation}\label{ci}
G^{s,0}_{0,s}\left(\begin{array}{cccc}
                - \\
               \nu_1+j,\nu_2\ldots,\nu_s
             \end{array}\biggl|y\right)=\frac{1}{2\pi i}\int\limits_{c
             -i\infty}^{c
             +i\infty}
             \Gamma\left(u+\nu_1+j\right)\Gamma\left(u+\nu_2\right)\ldots\Gamma\left(u+\nu_s\right)
            y^{-u}du,
\end{equation}
for $c>0$ and $\nu_j\geq0$ $\forall j$, leaving the poles of the Gamma-functions to the left of the contour.
Therefore, we can write
\begin{equation}
\begin{split}
&Q_{s,p}(y)=
\frac{1}{2\pi i}\int\limits_{c-i\infty}^{c+i\infty}
\Gamma\left(u+\nu_1\right)\ldots\Gamma\left(u+\nu_s\right)
\left(\sum\limits_{j=0}^p\frac{\left(-p\right)_j\left(u+\nu_1\right)_j}{\left(1+\nu_1\right)_jj!}
\right)
y^{-u}du.
\end{split}
\end{equation}
The sum inside the integral above can be written as Gauss' hypergeometric function,
\begin{equation}
\sum\limits_{j=0}^p\frac{\left(-p\right)_j\left(u+\nu_1\right)_j}{\left(1+\nu_1\right)_j}
\frac{1}{j!}=
{}_2F_1\left(-p,u+\nu_1;1+\nu_1;1\right).
\end{equation}
By the Chu-Vandermonde formula for the Gauss hypergeometric functions (see, for example, Ismail \cite[Section 1.4]{Ismail})
 $$
{}_2F_1\left(-p,u+\nu_1;1+\nu_1;1\right)=\frac{\left(1-u\right)_p}{\left(1+\nu_1\right)_p}
=\left(-1\right)^p\frac{\Gamma(u)}{\Gamma(u-p)}\frac{1}{\left(1+\nu_1\right)_p}.
$$
So we find
\begin{equation}\label{QINTEGRAL1}
\begin{split}
Q_{s,p}(y)=\frac{(-1)^p\Gamma(1+\nu_1)}{2\pi i\Gamma\left(1+\nu_1+p\right)}
\int\limits_{c-i\infty}^{c+i\infty}
\frac{\prod\limits_{j=0}^s\Gamma\left(u+\nu_j\right)
}{\Gamma(u-p)}y^{-u}du,
\end{split}
\end{equation}
where $s\in\left\{1,\ldots,m-1\right\}$.

To obtain an integral representation for the function $Q_{m,p}(y)$ we
proceed as follows. Recalling the definition \eqref{K-def2}
\begin{equation}\label{Krepresentation}
\int\limits_0^{+\infty}e^{-\frac{\beta}{t}-\gamma t}t^{\nu-1}dt=
2\left(\frac{\beta}{\gamma}\right)^{\frac{\nu}{2}}K_{\nu}\left(2\sqrt{\beta\gamma}\right),
\;\;\; \beta>0,\;\gamma>0,
\end{equation}
together with the contour integral representation for the corresponding Meijer $G$-function equation (\ref{ci}), applying Fubini's Theorem 
we obtain
\begin{equation}
\begin{split}
&\phi_{0,m}\left(j+1,y\right)=\frac{1}{2\pi i}\int\limits_{c-i\infty}^{c+i\infty}
\Gamma\left(u+\nu_1+j\right)\Gamma\left(u+\nu_2\right)\ldots\Gamma\left(u+\nu_{m-1}\right)
{2b^uy^{\frac{u}{2}}K_u\left(2by^{\frac{1}{2}}\right)}{y^{-u}}du.
\end{split}
\end{equation}
Inserting this into the formula  for $Q_{m,p}(y)$, equation (\ref{QSP}),  gives
\begin{equation}
\begin{split}
&Q_{m,p}(y)=\frac{1}{2\pi i}\int\limits_{c-i\infty}^{c+i\infty}
\prod\limits_{j=1}^{m-1}\Gamma\left(u+\nu_j\right)
{}_2F_1\left(-p,u+\nu_1;1+\nu_1;1\right)
{2b^uy^{-\frac{u}{2}}K_u\left(2by^{\frac{1}{2}}\right)}
du.
\end{split}
\end{equation}
By the same arguments as above, 
this formula can be rewritten as
\begin{equation}\label{QINTEGRAL2}
\begin{split}
Q_{m,p}(y)=\frac{(-1)^p\Gamma(1+\nu_1)}{2\pi i\Gamma\left(1+\nu_1+p\right)}
\int\limits_{c-i\infty}^{c+i\infty}
\frac{\prod_{j=0}^m\Gamma\left(u+\nu_j\right)
}{\Gamma(u-p)}
\frac{2b^uy^{-\frac{u}{2}}K_u\left(2by^{\frac{1}{2}}\right)}{\Gamma(u)}du.
\end{split}
\end{equation}
Now we are ready to prove Theorem \ref{TheoremDoubleIntegralRepresentationExactKernel}. In order to derive the formula for the correlation kernel stated in the Theorem \ref{TheoremDoubleIntegralRepresentationExactKernel} it is enough to represent the sum
$$
S_{n,m}(r,x;s,y;b)=\sum\limits_{p=0}^{n-1}\frac{\Gamma\left(\nu_1+p+1\right)}{\Gamma^2\left(\nu_1+1\right)p!}
P_{r,p}(x)Q_{s,p}(y)
$$
as a double contour integral. For this purpose use the contour integral representations for the functions
$P_{r,p}(x)$ and $Q_{s,p}(y)$ obtained above, equations (\ref{PINTEGRAL1}), (\ref{PINTEGRAL2}), (\ref{QINTEGRAL1}),
and (\ref{QINTEGRAL2}), and the formula derived in \cite[Eq.5.3]{KuijlaarsZhang}:
\begin{equation}\label{SumGammaFormula}
\sum\limits_{p=0}^{n-1}\frac{\Gamma(t-p)}{\Gamma(u-p)}=\frac{1}{u-t-1}\left[\frac{\Gamma(t-n+1)}{\Gamma(u-n)}
-\frac{\Gamma(t+1)}{\Gamma(u)}\right].
\end{equation}
Then  note that by the Residue Theorem  the second term in the right hand side of  equation (\ref{SumGammaFormula}) does not contribute
to $S_{n,m}(r,x;s,y;b)$. Here, we choose $c=\frac12$, and due to the shift in argument of the function $q_s(u,y;b)$ in equation \eqref{Functionq} this leads to the contour as stated in Theorem \ref{TheoremDoubleIntegralRepresentationExactKernel}.
\qed

\begin{Remarks}
We would like to point out that for the special case $m=2$ and $r=s=2$
 the kernel $K_{n,m}(r,x;s,y;b)$ we just determined in \eqref{FirstFormulaForK} can be shown to agree with the correlation kernel $K_N(x,y)$ obtained by the authors in  \cite[Theorem 3.2]{AkemannStrahov}, for a particular choice of $b$. This agreement is not surprising, in view of the identification made after \eqref{A2}.
Formula \eqref{FirstFormulaForK} for our kernel gives
\begin{equation}\label{A0}
K_{n,m=2}(r=2,x;s=2,y;b)=\sum\limits_{p=0}^{n-1}\frac{\Gamma\left(\nu_1+p+1\right)}{\Gamma^2\left(\nu_1+1\right)p!}
P_{r=2,p}(x)Q_{s=2,p}(y).
\end{equation}
Equation \eqref{Pmn} implies that the function $P_{r=2,p}(x)$ can be written as
\begin{equation}\label{A1p}
P_{r=2,p}(x)=\sum\limits_{j=0}^p\frac{(-p)_j}{\left(\nu_1+1\right)_{j}}\frac{x^{\frac{j}{2}}I_j\left(2bx^{\frac{1}{2}}\right)}{b^jj!}.
\end{equation}
Now let us find a convenient representation for the function $Q_{s=2,p}(x)$. Equation \eqref{QSP} reads
\begin{equation}\label{A2q}
Q_{s=2,p}(x)=\sum\limits_{j=0}^p\phi_{0,s=2}(j+1,y)\frac{(-p)_j}{\left(\nu_1+1\right)_jj!}.
\end{equation}
The function $\phi_{0,s=2}(j+1,y)$ was obtained in Proposition \ref{PropositionKV}, see equation \eqref{phi0siy}, and take into account
that $s=m=2$.
Using the representation of the exponential function in terms of the Meijer $G$-function
\begin{equation}
\label{Gexp}
x^{\nu}
e^{-\frac{y}{x}}=G_{0,1}^{1,0}\left(\begin{array}{c}
                         - \\
                         \nu
                       \end{array}
\biggl|{x}\right),
\end{equation}
we have
\begin{equation}\label{A3}
\phi_{0,s=2}(j+1,y)=\int\limits_{0}^{\infty}G_{0,1}^{1,0}\left(\begin{array}{c}
                                                                 - \\
                                                                 \nu_1+j
                                                               \end{array}
\biggl|t\right)e^{-\frac{y}{t}-b^2t}\frac{dt}{t}
=2\left(\frac{y}{b^2+1}\right)^{\frac{\nu_1+j}{2}}K_{\nu_1+j}\left(2\sqrt{\left(b^2+1\right)y}\right).
\end{equation}
Here, we have used \eqref{K-def2}.
If we insert (\ref{A3}) in formula (\ref{A2q}), we  find
\begin{equation}\label{A5}
Q_{s=2,p}(x)=2\sum\limits_{j=0}^p\frac{(-p)_j}{\left(\nu_1+1\right)_jj!}\left(\frac{y}{b^2+1}\right)^{\frac{\nu_1+j}{2}}
K_{\nu_1+j}\left(2\sqrt{\left(1+b^2\right)y}\right).
\end{equation}
Thus the kernel $K_{n,2}(2,x;2,y;b)$ is given by equation (\ref{A0}), where the functions
$P_{2,p}(x)$ and $Q_{2,p}(y)$ are defined by equations (\ref{A1p}) and (\ref{A5}), respectively.

Next, we turn to the kernel $K_N(x,y)$ obtained in \cite[Theorem 3.2]{AkemannStrahov}. There, the parameters $\alpha(\mu)$ and $\delta(\mu)$ were defined as
$$
\alpha(\mu)=\frac{1+\mu}{2\mu},\;\;\;\delta(\mu)=\frac{1-\mu}{2\mu},
$$
where $\mu$ takes values in the interval $(0,1]$. We thus have
$$
\frac{\alpha(\mu)^2-\delta(\mu)^2}{\delta(\mu)}=\frac{2}{1-\mu},\;\;\; \frac{\alpha(\mu)^2-\delta(\mu)^2}{\alpha(\mu)}=\frac{2}{1+\mu},\;\;\; \alpha(\mu)^2-\delta(\mu)^2=\frac{1}{\mu}.
$$
Taking this into account, and identifying $\nu_1=\nu$ there, we see that the correlation kernel $K_N(x,y)$  in  \cite[Theorem 3.2]{AkemannStrahov} can be written as
\begin{equation}
K_N(x,y)=\sum\limits_{n=0}^{N-1}P_n(x)Q_n(y),
\end{equation}
where
\begin{equation}
P_n(x)=\frac{(-1)^n(\nu_1+n)!n!}{\sqrt{\mu}}\sum\limits_{k=0}^{n}\frac{(-n)_k}{(\nu_1+k)!k!}\left(\frac{2x^{\frac{1}{2}}}{1-\mu}\right)^k
I_k\left(\frac{1-\mu}{\mu}x^{\frac{1}{2}}\right),
\end{equation}
and
\begin{equation}
Q_n(y)=\frac{2(-1)^n}{\sqrt{\mu}(n!)^2}\sum\limits_{l=0}^{n}\frac{(-n)_l}{(\nu_1+l)!l!}\left(\frac{2y^{\frac{1}{2}}}{1+\mu}\right)^{l+\nu_1}
K_{l+\nu_1}\left(\frac{1+\mu}{\mu}y^{\frac{1}{2}}\right).
\end{equation}
Setting
\begin{equation}\label{A6}
b=\frac{1-\mu}{2\sqrt{\mu}},\;\;\; x=\frac{\zeta}{\mu},\;\;\; y=\frac{\eta}{\mu},
\end{equation}
equations (\ref{A1p}) and (\ref{A5})  take the form
\begin{equation}\label{A11}
P_{r=2,p}(x)=\sum\limits_{j=0}^p\frac{(-p)_j}{\left(\nu_1+1\right)_{j}j!}
\left(\frac{2\zeta^{\frac{1}{2}}}{1-\mu}\right)^j
I_j\left(\frac{1-\mu}{\mu}\zeta^{\frac{1}{2}}\right),
\end{equation}
and
\begin{equation}\label{A12}
Q_{s=2,p}(y)=2\sum\limits_{j=0}^p\frac{(-p)_j}{\left(\nu_1+1\right)_{j}j!}
\left(\frac{2\eta^{\frac{1}{2}}}{1+\mu}\right)^{j+\nu_1}
K_{j+\nu_1}\left(\frac{1+\mu}{\mu}\eta^{\frac{1}{2}}\right).
\end{equation}
Here, we used that $(\nu_1+1)_j=(\nu_1+j)!/\nu_1!$. Thus we obtain the following identity
\begin{equation}
\mu K_{n,m=2}\left(r=2,\frac{\zeta}{\mu};s=2,\frac{\eta}{\mu};\frac{1-\mu}{2\sqrt{\mu}}\right)=K_{N=n}\left(\zeta,\eta\right).
\end{equation}
The extra factor of $\mu$ in front of the kernel on the left-hand side is due to the
change of variables defined by equation (\ref{A6}).
\end{Remarks}

\section{Proof of Theorem \ref{TheoremHardEdgeGinibreCouplingProcess}}
Given the double contour integral representation for correlation kernel
$K_{n,m}\left(r,x;s,y;b\right)$, we are ready to study
the asymptotic of the Ginibre product process with coupling, i.e. the asymptotic of the  matrix product process
associated with the multi-matrix model (\ref{MainProbabilityMeasure1}).
Recall that we consider the coupling parameter $b$ as a function of $n$, and investigate
the hard edge scaling limit at the origin  in different asymptotic regimes.
\subsection{The weak coupling regime}\label{SectionWCR}
In this regime we assume that ${b(n)}/{\sqrt{n}}\rightarrow 0$ as $n\rightarrow\infty$.
Let us first establish equation (\ref{WeakCouplingRegimeLimit}). Consider the first term in the right-hand side of
equation (\ref{CorrelationKernelFormulaContour}). We will show that in the weak coupling regime
the hard edge scaling limit of the first term is given by
\begin{equation}\label{ftl}
\underset{n\rightarrow\infty}{\lim}\frac{1}{n}\phi_{r,s}\left(\frac{x}{n},\frac{y}{n};b(n)\right)
=\frac{1}{x}G^{s-r,0}_{0,s-r}\left(
\begin{array}{c}
  - \\
  \nu_{r+1},\ldots,\nu_s
\end{array}\biggl|\frac{y}{x}
\right)\textbf{1}_{s>r}.
\end{equation}
The explicit formula for $\phi_{r,s}\left(x,y;b\right)$ is given in the statement of Theorem \ref{TheoremDoubleIntegralRepresentationExactKernel}.
Using this explicit formula, we immediately see that equation (\ref{ftl}) holds true
for $1\leq r<s\leq m-1$, and for $1\leq s\leq r\leq m$. For $r=m-1$, $s=m$ the left
hand side of equation (\ref{ftl}) is equal to $\frac{1}{x}e^{-\frac{y}{x}}$, which can be rewritten as
$$
\frac{1}{x}G^{1,0}_{0,1}\left(
\begin{array}{c}
  - \\
  \nu_m
\end{array}\biggl|\frac{y}{x}
\right)=x^{\nu_m-1} e^{-\frac{y}{x}}
$$
(recall that $\nu_m=0$). Therefore, equation (\ref{ftl}) holds true for
$r=m-1$, $s=m$ as well. For $1\leq r\leq m-2$, $s=m$ we can write
\begin{equation}
\begin{split}
\frac{1}{n}\phi_{r,s}\left(\frac{x}{n},\frac{y}{n};b(n)\right)
=\frac{1}{x}\int\limits_0^{\infty}G^{m-r-1,0}_{0,m-r-1}\left(
\begin{array}{c}
  - \\
  \nu_{r+1},\ldots,\nu_{m-1}
\end{array}\biggl|\frac{nt}{x}
\right)e^{-\frac{y}{nt}-b^2(n)t}\frac{dt}{t}.
\end{split}
\end{equation}
Changing the integration variable as $t=\frac{\tau}{n}$, and taking into account
that $e^{-\frac{b^2(n)}{n}\tau}\rightarrow 1$, in the weak coupling regime
we obtain that equation  (\ref{ftl}) holds true for  $1\leq r\leq m-2$, $s=m$.
In this calculations the procedure of taking the limit inside the integral can be justified by the dominated convergence theorem.

It remains to show that
\begin{equation}\label{WeakCouplingRegimeLimit1}
\begin{split}
&\underset{n\rightarrow\infty}{\lim}\left\{\frac{1}{n}S_{n,m}\left(r,\frac{x}{n};s,\frac{y}{n};b(n)\right)\right\}\\
&=\frac{1}{(2\pi i)^2}\int\limits_{-\frac{1}{2}-i\infty}^{-\frac{1}{2}+i\infty}
du\oint\limits_{\Sigma_{\infty}}dt\frac{\prod_{j=0}^s\Gamma(u+\nu_j+1)}{\prod_{j=0}^r\Gamma(t+\nu_j+1)}
\frac{\sin\pi u}{\sin\pi t}
\frac{x^ty^{-u-1}}{u-t},
\end{split}
\end{equation}
where $\Sigma_{\infty}$ is a contour starting from $+\infty$ in the upper half plane and returning to $+\infty$
in the lower half plane, leaving $-\frac{1}{2}$ on the left, and encircling $\left\{0,1,2,\ldots\right\}$.
To see that equation (\ref{WeakCouplingRegimeLimit1}) indeed holds true observe that
for $r,s\in\left\{1,\ldots,m-1\right\}$ the functions $p_r\left(t,\frac{x}{n};b(n)\right)$
and $q_s\left(u,\frac{y}{n};b(n)\right)$ are both identically equal to $1$, and that
for $r,s\in\left\{1,\ldots,m-1\right\}$ equation (\ref{WeakCouplingRegimeLimit}) can be obtained
by applying the same arguments as in Kuijlaars and Zhang \cite[Section 5.2.]{KuijlaarsZhang}.
In addition, we have
\begin{equation}
p_m\left(t,\frac{x}{n};b(n)\right)=\frac{\Gamma(t+1)
I_t\left(2\frac{b(n)}{\sqrt{n}}x^{\frac{1}{2}}\right)}{\left(\frac{b(n)}{\sqrt{n}}x^{\frac{1}{2}}\right)^t}\simeq 1;
\mbox{as}\; n\rightarrow\infty, \; \mbox{and}\;\frac{b(n)}{n^{\frac{1}{2}}}\rightarrow 0,
\end{equation}
and
\begin{equation}
q_m\left(u+1,\frac{y}{n};b(n)\right)=\frac{2\left(\frac{b(n)}{\sqrt{n}}y^{\frac{1}{2}}\right)^u
K_u\left(2\frac{b(n)}{\sqrt{n}}y^{\frac{1}{2}}\right)}{\Gamma(u)}\simeq 1;
\mbox{as}\; n\rightarrow\infty, \; \mbox{and}\;\frac{b(n)}{n^{\frac{1}{2}}}\rightarrow 0.
\end{equation}
Therefore, equation (\ref{WeakCouplingRegimeLimit1}) remains true for all
$r,s\in\left\{1,\ldots,m\right\}$. In these calculations the procedure of taking the limit inside
the double integral can be justified as in the proof of Theorem 5.3 in Kujlaars and Zhang \cite{KuijlaarsZhang}
using the dominated convergence theorem, and the asymptotic properties of the involved
Gamma and Bessel functions.

Now, the correlation kernel
of
the scaled Ginibre product process with coupling, formed by the configurations
$\left(ny_1^m,\ldots,ny_n^m;\ldots,ny_1^1,\ldots,ny_n^1\right)$, is related to
the correlation kernel
of the unscaled process
with configurations
$\left(y_1^m,\ldots,y_n^m;\ldots,y_1^1,\ldots,y_n^1\right)$
by replacing $K_{n,m}\left(r,x;s,y;b(n)\right)$
with
$
\frac{1}{n}K_{n,m}\left(r,\frac{x}{n};s,\frac{y}{n};b(n)\right).
$
Taking this into account, we see that  equation (\ref{WeakCouplingRegimeLimit})
implies the statement of Theorem \ref{TheoremHardEdgeGinibreCouplingProcess} in the weak coupling
regime.
\subsection{The intermediate coupling regime}
The proof of Theorem \ref{TheoremHardEdgeGinibreCouplingProcess} for
the intermediate coupling regime is based on arguments similar to those
presented in Section \ref{SectionWCR}. The difference is that in the intermediate
coupling regime we have for the transition functions that depend on $b(N)$, that is for $r=m$ or $s=m$:
\begin{equation}
\begin{split}
\underset{n\rightarrow\infty}{\lim}\left(\frac{1}{n}\phi_{r,s}\left(\frac{x}{n},\frac{y}{n};b(n)\right)\right)
=\phi_{r,s}\left(x,y;\alpha\right).
\end{split}
\nonumber
\end{equation}
This can be seen immediately from the formulae for $\phi_{r,s}\left(x,y;b(n)\right)$
in the statement of Theorem \ref{TheoremDoubleIntegralRepresentationExactKernel} when $n\rightarrow\infty$ and $\frac{b(n)}{n^{\frac{1}{2}}}\rightarrow\alpha$, as well as
\begin{equation}
\begin{split}
&p_m\left(t,\frac{x}{n};b(n)\right)\rightarrow p_m\left(t,x;\alpha\right),\\
&q_m\left(u,\frac{y}{n};b(n)\right)\rightarrow q_m\left(u,y;\alpha\right).
\end{split}
\nonumber
\end{equation}
using \eqref{Functionp} and \eqref{Functionq}.

\subsection{The strong coupling regime}
Let us first establish the limiting relations for the correlation
kernel in the strong coupling regime, equations (\ref{SKR2})-(\ref{SKR4}).
The kernels
with levels $1\leq r,s\leq m-1$ are $b$-independent, and the limit was already established in \eqref{WeakCouplingRegimeLimit}.

In order to establish equations (\ref{SKR2})-(\ref{SKR4}) we need to take into account
that $p_m(t,x;b)$ and $q_m(u,y;b)$ are nontrivial expressions involving the
modified Bessel functions of the first and of the second kind,
see equations (\ref{Functionp}), (\ref{Functionq}). These functions
have the following asymptotic as ${b(n)}/{n^{\frac{1}{2}}}\rightarrow\infty$:
\begin{equation}
p_m\left(t,\frac{b^2(n)x^2}{n^2};b(n)\right)\sim\frac{\Gamma(t+1)}{2\pi^{\frac{1}{2}}}
\frac{e^{2\frac{b^2(n)}{n}x}}{\left(\frac{b^2(n)}{n}x\right)^{t+\frac{1}{2}}},
\end{equation}
and
\begin{equation}
q_m\left(u+1,\frac{b^2(n)y^2}{n^2};b(n)\right)\sim\frac{\pi^{\frac{1}{2}}}{\Gamma(u+1)}
\frac{b^{2u+1}(n)}{n^{u+\frac{1}{2}}}y^{u+\frac{1}{2}}
e^{-2\frac{b^2(n)}{n}y}.
\end{equation}
The asymptotic formulae just written above follow immediately from the known
asymptotic of the modified Bessel function of the first kind $I_t\left(2bx^{\frac{1}{2}}\right)$, and
from the known asymptotic of the modified Bessel function of the second kind $K_u\left(2by^{\frac{1}{2}}\right)$
as $b\rightarrow\infty$, see equations (\ref{A1}) and (\ref{A2}), respectively. In addition, we have the following
ratio asymptotic of Gamma functions
$$
\frac{\Gamma(t-n+1)}{\Gamma(u-n+1)}=\frac{\sin\pi u}{\sin\pi t}n^{t-u}\left(1+O\left(n^{-1}\right)\right),\;\;\mbox{as}\;\; n\rightarrow\infty.
$$
Using the asymptotic formulae mentioned above we obtain the following limiting relations.\\
$\bullet$
For $1\leq r\leq m-1$, $s=m$ we have
\begin{equation}\label{SKR22}
\begin{split}
&\underset{n\rightarrow\infty}{\lim}\frac{2^{\frac{1}{2}}b(n)y^{\frac{1}{2}}}{n^{\frac{3}{2}}}
S_{n,m}\left(r,\frac{x}{n};m,\frac{b^2(n)y^2}{n^2};b(n)\right)\frac{e^{2\frac{b^2(n)}{n}y}}{2^{\frac{1}{2}}\pi^{\frac{1}{2}}}\\
&=\frac{1}{(2\pi i)^2}\int\limits_{-\frac{1}{2}-i\infty}^{-\frac{1}{2}+i\infty}
du\oint\limits_{\Sigma_{\infty}}dt\frac{\prod_{j=0}^{m-1}\Gamma(u+\nu_j+1)}{\prod_{j=0}^r\Gamma(t+\nu_j+1)}
\frac{\sin\pi u}{\sin\pi t}
\frac{x^ty^{-u-1}}{u-t}.
\end{split}
\end{equation}
$\bullet$
For $r=m$, $1\leq s\leq m-1$  we have
\begin{equation}\label{SKR33}
\begin{split}
&\underset{n\rightarrow\infty}{\lim}\frac{2^{\frac{1}{2}}b(n)x^{\frac{1}{2}}}{n^{\frac{3}{2}}}
S_{n,m}\left(m,\frac{b^2(n)x^2}{n^2};s,\frac{y}{n};b(n)\right)\frac{2^{\frac{1}{2}}\pi^{\frac{1}{2}}}{e^{2\frac{b^2(n)}{n}x}}\\
&=\frac{1}{(2\pi i)^2}\int\limits_{-\frac{1}{2}-i\infty}^{-\frac{1}{2}+i\infty}
du\oint\limits_{\Sigma_{\infty}}dt\frac{\prod_{j=0}^s\Gamma(u+\nu_j+1)}{\prod_{j=0}^{m-1}\Gamma(t+\nu_j+1)}
\frac{\sin\pi u}{\sin\pi t}
\frac{x^ty^{-u-1}}{u-t}.
\end{split}
\end{equation}
$\bullet$
Finally, for $r=m$ and $s=m$ the following limiting relation holds true
\begin{equation}\label{SKR44}
\begin{split}
&\underset{n\rightarrow\infty}{\lim}\frac{2b^2(n)x^{\frac{1}{2}}y^{\frac{1}{2}}}{n^{2}}
S_{n,m}\left(m,\frac{b^2(n)x^2}{n^2};m,\frac{b^2(n)y^2}{n^2};b(n)\right)
\frac{e^{2\frac{b^2(n)}{n}y}}{e^{2\frac{b^2(n)}{n}x}}\\
&=\frac{1}{(2\pi i)^2}\int\limits_{-\frac{1}{2}-i\infty}^{-\frac{1}{2}+i\infty}
du\oint\limits_{\Sigma_{\infty}}dt\frac{\prod_{j=0}^{m-1}\Gamma(u+\nu_j+1)}{\prod_{j=0}^{m-1}\Gamma(t+\nu_j+1)}
\frac{\sin\pi u}{\sin\pi t}
\frac{x^ty^{-u-1}}{u-t}.
\end{split}
\end{equation}
In the derivation of these formulae the interchange of limits and integrals can be justified by the dominated convergence theorem
in the same way as it is done in Kuijlaars and Zhang \cite[Section 5.2.]{KuijlaarsZhang}.
Formulae (\ref{SKR22})-(\ref{SKR44}) give the scaling limits of the second term,
$S_{n,m}\left(r,x;s,y;b(n)\right)$, in equation (\ref{CorrelationKernelFormulaContour})
for the correlation kernel $K_{n,m}\left(r,x;s,y;b(n)\right)$.

Now, let us find the asymptotic
of the first term, $\phi_{r,s}(x,y;b)$, in equation (\ref{CorrelationKernelFormulaContour}).
Recall that we already know the  asymptotic
of  $\phi_{r,s}(x,y;b)$ for $1\leq r,s\leq m-1$. So it is enough to consider
two cases.

$\bullet$ The first case corresponds to $1\leq r\leq m-2$, $s=m$.
In this case we have
\begin{equation}
\begin{split}
&\frac{2^{\frac{1}{2}}b(n)y^{\frac{1}{2}}}{n^{\frac{3}{2}}}
\phi_{r,m}\left(\frac{x}{n};\frac{b^2(n)y^2}{n^2};b(n)\right)\frac{e^{2\frac{b^2(n)}{n}y}}{2^{\frac{1}{2}}\pi^{\frac{1}{2}}}\\
&=\frac{2^{\frac{1}{2}}b(n)y^{\frac{1}{2}}}{n^{\frac{3}{2}}}\frac{n}{x}
\left[\int\limits_0^{\infty}G_{0,m-r-1}^{m-r-1,0}\left(\begin{array}{c}
                                                         - \\
                                                         \nu_{r+1},\ldots,\nu_{m-1}
                                                       \end{array}
\biggl|\frac{nt}{x}\right)e^{-\frac{b^2(n)y^2}{n^2t}-b^2(n)t}\frac{dt}{t}\right]\frac{e^{2\frac{b^2(n)}{n}y}}{2^{\frac{1}{2}}\pi^{\frac{1}{2}}}.
\end{split}
\end{equation}
The expression in the brackets can be written as
$$
I(\Lambda)=\int\limits_0^{\infty}G_{0,m-r-1}^{m-r-1,0}\left(\begin{array}{c}
                                                         - \\
                                                         \nu_{r+1},\ldots,\nu_{m-1}
                                                       \end{array}
\biggl|\frac{\tau}{x}\right)e^{-\Lambda\left(\frac{y^2}{\tau}+\tau\right)}
\frac{d\tau}{\tau},
$$
where $\Lambda=\frac{b^2(n)}{n}$. Using  Laplace's method for asymptotic expansions of integrals (see, for example, Miller \cite[Chapter 3.]{Miller}) we obtain that
the expression in the brackets has the following asymptotic
\begin{equation}
\begin{split}
&\int\limits_0^{\infty}G_{0,m-r-1}^{m-r-1,0}\left(\begin{array}{c}
                                                         - \\
                                                         \nu_{r+1},\ldots,\nu_{m-1}
                                                       \end{array}
\biggl|\frac{nt}{x}\right)e^{-\frac{b^2(n)y^2}{n^2t}-b^2(n)t}\frac{dt}{t}\\
&\qquad\qquad\sim\frac{\pi^{\frac{1}{2}}n^{\frac{1}{2}}}{b(n)y^{\frac{1}{2}}}
G_{0,m-r-1}^{m-r-1,0}\left(\begin{array}{c}
                                                         - \\
                                                         \nu_{r+1},\ldots,\nu_{m-1}
                                                       \end{array}
\biggl|\frac{y}{x}\right)e^{-\frac{2b^2(n)y}{n}},\;\;\mbox{as}\;\; \frac{b(n)}{\sqrt{n}}\rightarrow\infty.
\end{split}
\end{equation}
This gives
\begin{equation}\label{philimit1}
\underset{n\rightarrow\infty}{\lim}\frac{2^{\frac{1}{2}}b(n)y^{\frac{1}{2}}}{n^{\frac{3}{2}}}
\phi_{r,m}\left(\frac{x}{n};\frac{b^2(n)y^2}{n^2};b(n)\right)\frac{e^{2\frac{b^2(n)}{n}y}}{2^{\frac{1}{2}}\pi^{\frac{1}{2}}}
=\frac{1}{x}
G_{0,m-r-1}^{m-r-1,0}\left(\begin{array}{c}
                                                         - \\
                                                         \nu_{r+1},\ldots,\nu_{m-1}
                                                       \end{array}
\biggl|\frac{y}{x}\right),
\end{equation}
where $1\leq r\leq m-2$.

$\bullet$ For $r=m-1$, $s=m$ we can write
\begin{equation}
\begin{split}
\frac{2^{\frac{1}{2}}b(n)y^{\frac{1}{2}}}{n^{\frac{3}{2}}}
\phi_{m-1,m}\left(\frac{x}{n};\frac{b^2(n)y^2}{n^2};b(n)\right)\frac{e^{2\frac{b^2(n)}{n}y}}{2^{\frac{1}{2}}\pi^{\frac{1}{2}}}
&=\frac{b(n)y^{\frac{1}{2}}}{\pi^{\frac{1}{2}}n^{\frac{3}{2}}}\frac{n}{x}e^{-2\frac{b^2(n)y^2}{nx}-\frac{b^2(n)x}{n}+\frac{2b^2(n)y}{n}}\\
&=\frac{b(n)y^{\frac{1}{2}}}{\pi^{\frac{1}{2}}n^{\frac{1}{2}}x}e^{-\frac{b^2(n)}{xn}\left(y-x\right)^2}.
\end{split}
\end{equation}
We conclude that
\begin{equation}\label{philimit2}
\begin{split}
\underset{n\rightarrow\infty}{\lim}\frac{2^{\frac{1}{2}}b(n)y^{\frac{1}{2}}}{n^{\frac{3}{2}}}
\phi_{m-1,m}\left(\frac{x}{n};\frac{b^2(n)y^2}{n^2};b(n)\right)\frac{e^{2\frac{b^2(n)}{n}y}}{2^{\frac{1}{2}}\pi^{\frac{1}{2}}}=\delta(x-y),
\end{split}
\end{equation}
representing the Dirac delta distribution.
Taking into account equations (\ref{SKR22})-(\ref{SKR44}), (\ref{philimit1}), and (\ref{philimit2})
we obtain the limiting relations (\ref{SKR2})-(\ref{SKR4}) in the statement of Theorem \ref{TheoremHardEdgeGinibreCouplingProcess}.

In the weak and intermediate coupling regime in the previous two subsections the points on all levels were rescaled in the same way with $n$. In the strong coupling regime this is no longer so. Let us therefore explain how the different rescaling of the point configurations translates into different rescalings of the various correlation kernels. This justifies the different scalings that we have applied above.
Assume that the point configurations
$$
\left(y_1^m,\ldots,y_n^m;\ldots; y_1^1,\ldots,y_n^1\right)
$$
form a determinantal process on $\left\{1,\ldots,m\right\}\times\R_{>0}$
defined by the correlation kernel $K_{n,m}\left(r,x;s,y;b\right)$. Consider the scaled determinantal process on
$\left\{1,\ldots,m\right\}\times\R_{>0}$ formed by the scaled point configurations
\begin{equation}
\begin{split}
&\left(\frac{n}{b(n)}\left(y_1^m\right)^{\frac{1}{2}},\ldots,\frac{n}{b(n)}\left(y_n^m\right)^{\frac{1}{2}};
ny_1^{m-1},\ldots,ny_n^{m-1};\ldots;ny_1^{m-1},\ldots,ny_n^{m-1}\right)\\
&=\left(u_1^m,\ldots,u_n^m;u_1^{m-1},\ldots,u_n^{m-1};\ldots; u_1^1,\ldots,u_n^1\right).
\end{split}
\nonumber
\end{equation}
If $\varrho_{k_1,\ldots,k_m}$ is the correlation function of the original determinantal process, see their definition in \eqref{rhok-def}, and
$\widehat{\varrho}_{k_1,\ldots,k_m}$ is that of the scaled determinantal process, then we must have
\begin{equation}
\label{rhokscale}
\begin{split}
&\varrho_{k_1,\ldots,k_m}\left(y_1^1,\ldots,y_{k_1}^1;\ldots;y_1^m,\ldots,y_{k_m}^m\right)
\left(dy_1^1,\ldots,dy_{k_1}^1\right)\ldots\left(dy_1^m,\ldots,dy_{k_m}^m\right)\\
&=\widehat{\varrho}_{k_1,\ldots,k_m}\left(u_1^1,\ldots,u_{k_1}^1;\ldots;u_1^m,\ldots,u_{k_m}^m\right)
\left(du_1^1,\ldots,du_{k_1}^1\right)\ldots\left(du_1^m,\ldots,du_{k_m}^m\right).
\end{split}
\end{equation}
This gives the following relation between the initial correlation function, $\varrho_{k_1,\ldots,k_m}$,
and the correlation function $\widehat{\varrho}_{k_1,\ldots,k_m}$ of the scaled determinant process
\begin{equation}\label{RelationsBetweenCorrelationFunctions}
\begin{split}
&\widehat{\varrho}_{k_1,\ldots,k_m}\left(u_1^1,\ldots,u_{k_1}^1;\ldots;u_1^{m-1},\ldots,u_{k_{m-1}}^{m-1};
u_1^m,\ldots,u_{k_m}^m\right)\\
&=\frac{2^{k_m}u_1^m\ldots u_{k_m}^m}{n^{k_1+\ldots+k_{m-1}}}\left(\frac{b(n)}{n}\right)^{2k_m}\\
&\times\varrho_{k_1,\ldots,k_m}\left(\frac{u_1^1}{n},\ldots,\frac{u_{k_1}^1}{n};\ldots;\frac{u_1^{m-1}}{n},\ldots,\frac{u_{k_{m-1}}^{m-1}}{n};
\frac{b^2(n)}{n^2}\left(u_1^m\right)^2,\ldots,\frac{b^2(n)}{n^2}\left(u_{k_m}^m\right)^2\right).
\end{split}
\end{equation}
If $K_{n,m}\left(r,u;s,v\right)$ is the correlation kernel of the initial determinantal process, and
$\widehat{K}_{n,m}\left(r,u;s,v\right)$ is the correlation kernel of the scaled determinantal process,
then relation (\ref{RelationsBetweenCorrelationFunctions}) between the correlation functions implies
\begin{equation}\label{k0}
\begin{split}
&\widehat{K}_{n,m}\left(r,u;s,v\right)\\
&=\left\{
    \begin{array}{ll}
      \frac{1}{n}K_{n,m}\left(r,\frac{u}{n};s,\frac{v}{n}\right), & 1\leq r,s\leq m-1, \\
      \frac{2^{\frac{1}{2}}b(n)v^{\frac{1}{2}}}{n^{\frac{3}{2}}}K_{n,m}\left(r,\frac{u}{n};m,\frac{b^2(n)}{n^2}v^2\right), & 1\leq r\leq m-1, s=m, \\
      \frac{2^{\frac{1}{2}}b(n)u^{\frac{1}{2}}}{n^{\frac{3}{2}}}K_{n,m}\left(m,\frac{b^2(n)}{n^2}u^2;s,\frac{v}{n}\right), & r=m, 1\leq s\leq m-1, \\
      \frac{2b^2(n)u^{\frac{1}{2}}v^{\frac{1}{2}}}{n^2}K_{n,m}\left(m,\frac{b^2(n)}{n^2}u^2;m,\frac{b^2(n)}{n^2}v^2\right), & r=m, s=m.
    \end{array}
  \right.
\end{split}
\end{equation}
Also, let us take into account that if we have two correlation
kernels, say $K_{n,m}\left(r,x;s,y\right)$ and $K_{n,m}'\left(r,x;s,y\right)$
living on $\left\{1,\ldots,m\right\}\times\R_{>0}$, and related as
$$
K_{n,m}'\left(r,x;s,y\right)=\frac{\varphi_r(x)}{\varphi_s(y)}K_{n,m}\left(r,x;s,y\right),
$$
where $\varphi_1$, $\ldots$, $\varphi_m$ are certain non-vanishing functions defined on $\R_{>0}$, then
both $K_{n,m}\left(r,x;s,y\right)$ and $K_{n,m}'\left(r,x;s,y\right)$ give
the same correlation functions, and therefore define the same determinantal point process
on $\left\{1,\ldots,m\right\}\times\R_{>0}$. The two kernels are therefore called equivalent kernels. In other words, we have from \eqref{rhokscale}
\begin{equation}
\begin{split}
&\det\left[\begin{array}{ccc}
             \left(K_{n,m}(1,x_i^1;1,x_j^1)\right)_{i=1,\ldots,k_1}^{j=1,\ldots,k_1}
              & \ldots & \left(K_{n,m}(1,x_i^1;m,x_j^m)\right)_{i=1,\ldots,k_1}^{j=1,\ldots,k_m} \\
              \vdots &  &  \\
 \left(K_{n,m}(m,x_i^m;1,x_j^1)\right)_{i=1,\ldots,k_m}^{j=1,\ldots,k_1}
  & \ldots & \left(K_{n,m}(m,x_i^m;m,x_j^m)\right)_{i=1,\ldots,k_m}^{j=1,\ldots,k_m}
            \end{array}
\right]\\
&=\det\left[\begin{array}{ccc}
             \left(K_{n,m}'(1,x_i^1;1,x_j^1)\right)_{i=1,\ldots,k_1}^{j=1,\ldots,k_1}
              & \ldots & \left(K_{n,m}'(1,x_i^1;m,x_j^m)\right)_{i=1,\ldots,k_1}^{j=1,\ldots,k_m} \\
              \vdots &  &  \\
 \left(K_{n,m}'(m,x_i^m;1,x_j^1)\right)_{i=1,\ldots,k_m}^{j=1,\ldots,k_1}
  & \ldots & \left(K_{n,m}'(m,x_i^m;m,x_j^m)\right)_{i=1,\ldots,k_m}^{j=1,\ldots,k_m}
            \end{array}
\right].\\
\end{split}
\nonumber
\end{equation}
Therefore, the correlation kernel $\widehat{K}'_{n,m}\left(r,x;s,y\right)$ defined by
\begin{equation}
\begin{split}
&\widehat{K}'_{n,m}\left(r,u;s,v\right)\\
&=\left\{
    \begin{array}{ll}
      \frac{1}{n}K_{n,m}\left(r,\frac{u}{n};s,\frac{v}{n}\right), & 1\leq r,s\leq m-1, \\
      \frac{2^{\frac{1}{2}}b(n)v^{\frac{1}{2}}}{n^{\frac{3}{2}}}K_{n,m}\left(r,\frac{u}{n};m,\frac{b^2(n)}{n^2}v^2\right)
      \frac{e^{2\frac{b^2(n)}{n}v}}{2^{\frac{1}{2}}\pi^{\frac{1}{2}}}, & 1\leq r\leq m-1, s=m, \\
      \frac{2^{\frac{1}{2}}b(n)u^{\frac{1}{2}}}{n^{\frac{3}{2}}}K_{n,m}\left(m,\frac{b^2(n)}{n^2}u^2;s,\frac{v}{n}\right)
      \frac{2^{\frac{1}{2}}\pi^{\frac{1}{2}}}{e^{2\frac{b^2(n)}{n}u}}, & r=m, 1\leq s\leq m-1, \\
      \frac{2b^2(n)u^{\frac{1}{2}}v^{\frac{1}{2}}}{n^2}K_{n,m}\left(m,\frac{b^2(n)}{n^2}u^2;m,\frac{b^2(n)}{n^2}v^2\right)
      \frac{e^{2\frac{b^2(n)}{n}v}}{e^{2\frac{b^2(n)}{n}u}}, & r=m, s=m,
    \end{array}
  \right.
\end{split}
\end{equation}
gives the same scaled determinantal process as the correlation kernel  $\widehat{K}_{n,m}\left(r,x;s,y\right)$ defined by
equation (\ref{k0}). We conclude that the limiting relations (\ref{SKR2})-(\ref{SKR4})
imply that the scaled determinantal process on $\left\{1,\ldots,m\right\}\times\R_{>0}$
formed by the scaled point configurations
\begin{equation}
\begin{split}
\left(\frac{n}{b(n)}\left(y_1^m\right)^{\frac{1}{2}},\ldots,\frac{n}{b(n)}\left(y_n^m\right)^{\frac{1}{2}};
ny_1^{m-1},\ldots,ny_n^{m-1};\ldots;ny_1^{m-1},\ldots,ny_n^{m-1}\right)
\end{split}
\nonumber
\end{equation}
converges as $n\rightarrow\infty$ to the determinantal point process $\Pc^{\Ginibre}_{\infty,m-1}$ on
$\left\{1,\ldots,m-1\right\}\times\R_{>0}$
with the identification of levels $m$ and $m-1$ as given in Theorem \ref{TheoremHardEdgeGinibreCouplingProcess} {(C)}.
\qed
\section{Proof of Theorem \ref{TheoremInterpolation}}\label{SectionProofTheoremInterpolation}
Consider the determinantal point process on $\left\{1,\ldots,m\right\}\times\R_{>0}$
defined by the correlation kernel $K_{\infty,m}^{\interpol}(r,x;s,y;\alpha)$, see
equation (\ref{KernelZKExtended}). From equation (\ref{PHI(x,y,b)}) we obtain that
$$
\underset{\alpha\rightarrow 0}{\lim}\phi_{r,s}(x,y;\alpha)=\frac{1}{x}
G_{0,s-r}^{s-r,0}\left(\begin{array}{c}
                         - \\
                         \nu_{r+1},\ldots,\nu_s
                       \end{array}
\biggl|\frac{y}{x}\right)\textbf{1}_{s>r}.
$$
Here, we have used the fact that $\nu_m=0$, the integration formula
\begin{equation}
\begin{split}
\int\limits_0^{\infty}G_{0,m-r-1}^{m-r-1,0}\left(\begin{array}{c}
                         - \\
                         \nu_{r+1},\ldots,\nu_{m-1}
                       \end{array}
\biggl|\frac{y}{t}\right)
t^{\nu_m}e^{-t}
\frac{dt}{t}
=G_{0,m-r}^{m-r,0}\left(\begin{array}{c}
                         - \\
                         \nu_{r+1},\ldots,\nu_{m}
                       \end{array}
\biggl|{y}\right),
\end{split}
\end{equation}
together with the representation of the exponential function in terms of the Meijer $G$-function \eqref{Gexp}.
Thus, as $\alpha\rightarrow 0$, the first term in the right-hand side
of equation (\ref{KernelZKExtended}) for the correlation kernel $K_{\infty,m}^{\interpol}(r,x;s,y;\alpha)$
turns into the first term in the right-hand side
of equation (\ref{KernelInfiniteGinibreProductProcess}) for the correlation kernel $K_{\infty,m}^{\Ginibre}(r,x;s,y)$.
Moreover, taking into account that
$$
\underset{\alpha\rightarrow 0}{\lim}p_m(t,x;\alpha)=1,\;\; \underset{\alpha\rightarrow 0}{\lim}q_m(u,y;\alpha)=1,
$$
from equations (\ref{ISmallb}) and (\ref{KSmallb}),
we see that the second term in formula (\ref{KernelZKExtended}) for $K_{\infty,m}^{\interpol}(r,x;s,y;\alpha)$
turns into the second term in formula (\ref{KernelInfiniteGinibreProductProcess}) for  $K_{\infty,m}^{\Ginibre}(r,x;s,y)$. In these calculations the procedure of taking the limit inside
the double integral can be justified
using the dominated convergence theorem.
Thus Theorem \ref{TheoremInterpolation} {(A)} is proved.

The proof of Theorem \ref{TheoremInterpolation} {(B)} is very similar to that of Theorem \ref{TheoremHardEdgeGinibreCouplingProcess} {(C)}.
For this reason we present the sketch of the proof of  Theorem \ref{TheoremInterpolation} {(B)} only.
For $1\leq r,s\leq m-1$ the kernel
does not depend on $\alpha$ and agrees with that of the Ginibre point process given by  equation (\ref{KernelInfiniteGinibreProductProcess}).
 In order to check equations (\ref{I2})-(\ref{I4}) we use the known asymptotic formulae
for the modified Bessel functions of the first and of the second kind (see equations (\ref{A1}) and (\ref{A2})) to deduce that
\begin{equation}
p_m\left(t,\alpha^2x^2;\alpha\right)\sim\frac{\Gamma(t+1)}{\left(\alpha^2x\right)^t}
\frac{e^{2\alpha^2x}}{2\pi^{\frac{1}{2}}\alpha x^{\frac{1}{2}}},\;\;\alpha\rightarrow\infty,
\end{equation}
and that
\begin{equation}
q_m\left(u+1,\alpha^2y^2;\alpha\right)\sim\frac{\pi^{\frac{1}{2}}}{\Gamma(u+1)}
\frac{\left(\alpha^2y\right)^{u+1}}{\alpha y^{\frac{1}{2}}}e^{-2\alpha^2y},\;\;\alpha\rightarrow\infty.
\end{equation}
Similar arguments as in the proof of Theorem \ref{TheoremHardEdgeGinibreCouplingProcess} {(C)},
can be applied to show that equations (\ref{I2})-(\ref{I4}) ensure the convergence of the scaled interpolating process to
$\mathcal{P}^{\Ginibre}_{\infty,m-1}$.\qed


\end{document}